\documentclass{article} % For LaTeX2e
\usepackage{graphicx}      % include this line if your document contains figures
\usepackage{cite}
\usepackage{graphicx}
\usepackage{algorithmic,algorithm}
\usepackage{amssymb,amsmath,amsthm}
\usepackage{mathrsfs}
\usepackage{xspace}
\usepackage{hyperref}
\usepackage{booktabs}
\usepackage{multicol}
\usepackage{bbm}

\usepackage{lindsten,abbreviations}
\renewcommand\nonnegatives{\reals_+} % \smaller not working in lindsten.sty !!!
\newcommand\pf{PF\@\xspace}
\newcommand\ffbsi{FFBS\@\xspace}  % Forward filter/backward simulator
\newcommand\acf{ACF\@\xspace}
\newcommand\pgas{PGAS\@\xspace}
\newcommand\pgbs{PGBS\@\xspace}
\newcommand\mcem{MCEM\@\xspace}
\newcommand\saem{SAEM\@\xspace}

\newcommand\TV{\operatorname{TV}}
\newcommand{\I}{\mathbbm{1}}
\newcommand\half{\textstyle\frac{1}{2}} % How to express 1/2 in all the exponentials
\newcommand\gibbsnum[1]{\emph{(#1)}}
\newcommand\normsum[1]{\sum_#1}    % Weight normalization sum over index #1

\newcommand\xx{\mathbf{x}}
\renewcommand\aa{\mathbf{a}}

\newcommand\extendedspace{\Omega}
\newcommand\tDist[2]{\gamma_{#1,#2}}
\newcommand\ntDist[2]{\bar\gamma_{#1,#2}}
\newcommand\etDist[2]{\widehat\gamma_{#1,#2}^\Np} % empirical target distribution
\newcommand\normconst[2]{Z_{#1,#2}}
\newcommand\M[2]{M_{#1,#2}}
\renewcommand\r[2]{r_{#1,#2}}
\newcommand\wf[2]{W_{#1,#2}}

\newcommand\kPGAS[1][\theta]{P_{#1}^\Np}
\newcommand\kPGBS[1][\theta]{P_{\mathsf{BS},#1}^\Np}
\newcommand\ergoR[1][\theta]{R_{#1}}
\newcommand\ergorho[1][\theta]{\rho_{#1}}
\newcommand\ergoepsilon[1][\theta]{\varepsilon_{#1}}
\newcommand\ergokappa[1][\theta]{\kappa_{#1}}
\newcommand\steplength{\alpha}
\newcommand\sampletime{dt}
\newcommand\populationsize{\mathcal{P}}
\DeclareMathOperator\logit{logit}
\newcommand\asprob{\rho}
\newcommand\asprobtrunc[1][\ell]{\widehat\rho_{#1}}
\newcommand\KLD{D_\mathsf{KLD}}
\newcommand\DTV{D_\mathsf{TV}}

\usepackage{color}

\newlength\Papproxplotheight
\newlength\boxplotheight
\newtheoremstyle{ex}% name of the style to be used
  { }% measure of space to leave above the theorem. E.g.: 3pt
  { }% measure of space to leave below the theorem. E.g.: 3pt
  { }% name of font to use in the body of the theorem
  { }% measure of space to indent
  {\bfseries}% name of head font
  { }% punctuation between head and body
  { }% space after theorem head; " " = normal interword space
  {\thmname{#1}\thmnumber{ #2}:\thmnote{ #3}}% Manually specify head
\theoremstyle{ex}

\theoremstyle{remark}

\theoremstyle{plain}
\newtheorem{definition}{Definition}
\newtheorem{theorem}{Theorem}
\newtheorem{lemma}{Lemma}

\newtheorem{proposition}{Proposition}
\newtheorem{proc}{Procedure}

\newcounter{asmp}
\def\asmpfont{\upshape}
\newenvironment{asmp}{\refstepcounter{asmp}\par\trivlist
   \item[\hskip \labelsep{\bfseries(A\theasmp)}]\asmpfont}
   {\endtrivlist}
\graphicspath{{./}{../graphics/}}

\title{Particle Gibbs with Ancestor Sampling}
\author{
  Fredrik Lindsten\\
  Div. of Automatic Control\\
  Link{\"o}ping University\\
  \texttt{lindsten@isy.liu.se} \\
  \and
  Michael I. Jordan\\
  Depts. of EECS and Statistics\\
  University of California, Berkeley\\
  \texttt{jordan@cs.berkeley.edu}
  \and
  Thomas B. Sch\"on \\
  Dept. of Information Technology\\
  Uppsala University\\
  \texttt{thomas.schon@it.uu.se} \\
}

\begin{document}
\maketitle

\begin{abstract}
% PMCMC
Particle Markov chain Monte Carlo (\pmcmc) is a systematic way of combining 
the two main tools used for Monte Carlo statistical inference: sequential 
Monte Carlo (\smc) and Markov chain Monte Carlo (\mcmc).
% PGAS -- new method
We present a novel \pmcmc algorithm that we refer to as \emph{particle 
Gibbs with ancestor sampling} (\pgas).  \pgas provides the data analyst
with an off-the-shelf class of Markov kernels that can be used to simulate
the typically high-dimensional and highly autocorrelated state trajectory 
in a state-space model.
% Small N
The \emph{ancestor sampling} procedure enables fast mixing of the \pgas kernel
even when using seemingly few particles in the underlying \smc sampler. 
This is important as it can significantly reduce the computational burden
that is typically associated with using \smc.
% Relationship to PGBS
\pgas is conceptually similar to the existing \emph{\pg with backward 
simulation} (\pgbs) procedure.  Instead of using separate forward and 
backward sweeps as in \pgbs, however, we achieve the same effect in a 
single forward sweep.
% SSMs -- Complex models
This makes \pgas well suited for addressing inference problems not only
in state-space models, but also in models with more complex dependencies, 
such as non-Markovian, Bayesian nonparametric, and general probabilistic 
graphical models.
\end{abstract}

%%%%%%%%%%%%%%%%%%%%%%%
%     NEW SECTION     %
%%%%%%%%%%%%%%%%%%%%%%%
\section{Introduction}\label{sec:intro}%
Monte Carlo methods are one of the standard tools for inference in statistical 
models as they, among other things, provide a systematic approach to the problem
of computing Bayesian posterior probabilities. Sequential Monte Carlo (\smc) 
\cite{DoucetJ:2011,DelMoralDJ:2006} and Markov chain Monte Carlo (\mcmc) 
\cite{RobertC:2004,Liu:2001} methods in particular have found application to 
a wide range of data analysis problems involving complex, high-dimensional
models.  These include state-space models (\ssm{s}) which are used in the 
context of time series and dynamical systems modeling in a wide range of 
scientific fields.  The strong assumptions of linearity and Gaussianity that 
were originally invoked for \ssm{s} have indeed been weakened by decades of 
research on \smc and \mcmc.

These methods have not, however, led to a substantial weakening of a further strong assumption, that of
Markovianity. It remains a major challenge to develop efficient inference algorithms for
models containing a latent stochastic process which, in contrast with the state process in an \ssm, is non-Markovian.
Such non-Markovian latent variable models arise in various settings, either from direct modeling or
via a transformation or marginalization of an \ssm. We discuss this further in Section~\ref{sec:nonmarkov}; see also \cite[Section~4]{LindstenS:2013}.

In this paper we present a new tool in the family of Monte Carlo methods which is particularly useful
for inference in \ssm{s} and, importantly, in non-Markovian latent variable models.
However, the proposed method is by no means limited to these model classes.
We work within the framework of particle \mcmc (\pmcmc) \cite{AndrieuDH:2010} which is
a systematic way of combining \smc and \mcmc, exploiting the strengths of both techniques.
More specifically, \pmcmc samplers make use of \smc to construct efficient, high-dimensional \mcmc kernels with certain
invariance properties. These kernels can then be used as off-the-shelf components in \mcmc algorithms and other inference
strategies relying on Markov kernels, such as Markovian stochastic approximation methods.
\pmcmc has, in a relatively short period of time, found many applications in areas such as hydrology \cite{VrugtBDS:2013},
finance \cite{PittSGK:2012}, systems biology \cite{GolightlyW:2011}, and epidemiology \cite{RasmussenRK:2011}, to mention a few.

Our method builds on the particle Gibbs (\pg) sampler proposed by \cite{AndrieuDH:2010}.
In \pg, the aforementioned Markov kernel is constructed by running an \smc sampler
in which one particle trajectory is set deterministically to 
a reference trajectory that is specified \emph{a priori}.
After a complete run of the \smc algorithm, a new trajectory is obtained by selecting one of the particle trajectories
with probabilities given by their importance weights. The effect of the reference trajectory
is that the resulting Markov kernel leaves its target distribution invariant, regardless of the number
of particles used in the underlying \smc algorithm.

However, \pg suffers from a serious drawback, which is that the mixing of the
Markov kernel can be very poor when there is path degeneracy in the underlying \smc sampler~\cite{LindstenS:2013,ChopinS:2013}.
Unfortunately, path degeneracy is inevitable for high-dimensional problems, which significantly reduces
the applicability of \pg. This problem has been addressed in the generic setting of \ssm{s} by adding a backward simulation 
step to the \pg sampler, yielding a method denoted as \emph{\pg with backward
simulation} (\pgbs) \cite{WhiteleyAD:2010,LindstenS:2012}.  It has been found that this considerably improves
mixing, making the method much more robust to a small number of particles
as well as growth in the size of the data \cite{LindstenS:2013,ChopinS:2013,WhiteleyAD:2010,LindstenS:2012}.

Unfortunately, however, the application of backward simulation is problematic
for models with more intricate dependencies than in \ssm{s}, such as non-Markovian latent variable models.
The reason is that we need to consider complete
trajectories of the latent process during the backward simulation pass
(see Section~\ref{sec:nonmarkov} for details). The method proposed in this paper,
which we refer to as \emph{particle Gibbs with ancestor sampling} (\pgas),
is geared toward this issue. \pgas alleviates the problem with path degeneracy by modifying the original \pg kernel with
a so called ancestor sampling step, thereby achieving the same effect as backward sampling, but without an explicit backward pass.

The \pgas Markov kernel is constructed in Section~\ref{sec:pgas}, extending the preliminary work that we have
previously published in \cite{LindstenJS:2012}.
It is also illustrated how ancestor sampling
can be used to mitigate the problems with path degeneracy which deteriorates the performance of \pg. In Section~\ref{sec:justification}
we establish the theoretical validity of the \pgas approach, including a novel uniform ergodicity result.
We then show specifically how \pgas can be used for inference and learning of \ssm{s} and of non-Markovian latent variable models in
Sections~\ref{sec:ssm}~and~\ref{sec:nonmarkov}, respectively.
The \pgas algorithm is then illustrated on several numerical examples in Section~\ref{sec:eval}.
As part of our development, we also propose a truncation strategy specifically for non-Markovian models.
This is a generic method that is also applicable to \pgbs, but, as we show in the simulation study in Section~\ref{sec:eval}, the effect of the
truncation error is much less severe for \pgas than for \pgbs.  Indeed, we obtain up to an order of magnitude increase
in accuracy in using \pgas when compared to \pgbs in this study.
Finally, in Section~\ref{sec:discussion} we conclude and 
point out possible directions for future work.

%%%%%%%%%%%%%%%%%%%%%%%
%     NEW SECTION     %
%%%%%%%%%%%%%%%%%%%%%%%
\section{Sequential Monte Carlo}\label{sec:smc}%
Let $\tDist{\theta}{t}(x_{1:t})$, for $t = \range{1}{\T}$, be a sequence of unnormalized densities\footnote{With respect to some dominating measure which we denote simply as $dx_{1:t}$.} on the measurable space $(\setX^t, \sigmaX^t)$,
parameterized by $\theta \in \Theta$.
Let $\ntDist{\theta}{t}(x_{1:t})$ be the corresponding
normalized probability densities:
\begin{align}
  \ntDist{\theta}{t}(x_{1:t}) = \frac{ \tDist{\theta}{t}(x_{1:t}) }{ \normconst{\theta}{t} },
\end{align}
where $\normconst{\theta}{t} = \int \ntDist{\theta}{t}(x_{1:t})\,dx_{1:t}$
and where it is assumed that $\normconst{\theta}{t} > 0, \forall\theta\in\Theta$.
For instance, in the (important) special case of an \ssm we have $\ntDist{\theta}{t}(x_{1:t}) = p_\theta(x_{1:t} \mid y_{1:t})$,
$\tDist{\theta}{t}(x_{1:t}) = p_\theta(x_{1:t}, y_{1:t})$, and $\normconst{\theta}{t} = p_\theta(y_{1:t})$.
We discuss this special case in more detail in Section~\ref{sec:ssm}.

To draw inference about the latent variables $x_{1:\T}$, as well as to enable learning of the model parameter $\theta$,
a useful approach is to construct a Monte Carlo algorithm to draw samples from $\ntDist{\theta}{\T}(x_{1:\T})$. The sequential nature
of the problem suggests the use of \smc methods; in particular, particle filters 
(\pf{s})~\cite{DoucetJ:2011,DelMoralDJ:2006,PittS:1999}.

We start by reviewing a standard \smc sampler, which will be used to construct the \pgas algorithm in the
consecutive section. We will refer to the index variable $t$ as time, but in general it might not have any temporal meaning.
Let $\{x_{1:t-1}^i, w_{t-1}^i\}_{i=1}^\Np$ be a weighted particle system targeting $\ntDist{\theta}{t-1}(x_{1:t-1})$.
That is, the weighted particles define an empirical point-mass approximation of the target distribution given by
\begin{align}
  \etDist{\theta}{t-1}(dx_{1:t-1}) = \sum_{i=1}^\Np \frac{ w_{t-1}^i}{ \sum_l w_{t-1}^l } \delta_{x_{1:t-1}^i}(dx_{1:t-1}).
\end{align}
This particle system is propagated to time $t$ by sampling $\{a_t^i, x_t^i\}_{i=1}^\Np$ independently from a proposal kernel,
\begin{align}
  \label{eq:smc:joint_proposal}
  \M{\theta}{t}(a_t, x_t) = \frac{w_{t-1}^{a_t} }{\normsum{l} w_{t-1}^l } \r{\theta}{t}(x_t \mid x_{1:t-1}^{a_t}).
\end{align}
Note that $\M{\theta}{t}$ depends on the complete particle system up to time $t-1$, $\{x_{1:t-1}^i, w_{t-1}^i\}_{i=1}^\Np$,
but for notational convenience we shall not make that dependence explicit.
Here, $a_t^i$ is the index of the ancestor particle of $x_t^i$.
In this formulation, the resampling step is implicit and corresponds to sampling these \emph{ancestor indices}.
When we write $x_{1:t}^i$ we refer to the ancestral path of particle $x_t^i$. That is, the particle trajectory is defined recursively as
\begin{align}
  \label{eq:smc:trajectory}
  x_{1:t}^i  = (x_{1:t-1}^{a_t^i}, x_t^i).
\end{align}
Once we have generated $\Np$ ancestor indices and particles from the proposal kernel \eqref{eq:smc:joint_proposal},
the particles are weighted according to $w_t^i = \wf{\theta}{t}(x_{1:t}^i)$ where the weight function is given by
\begin{align}
  \label{eq:smc:weightfunction}
  \wf{\theta}{t}(x_{1:t}) = \frac{\tDist{\theta}{t}(x_{1:t})}{ \tDist{\theta}{t-1}(x_{1:t-1}) \r{\theta}{t}(x_t \mid x_{1:t-1})},
\end{align}
for $t \geq 2$. The procedure is initialized by sampling from a proposal density $x_1^i \sim \r{\theta}{1}(x_1)$ and assigning importance weights
$w_1^i = \wf{\theta}{1}(x_1^i)$ with $\wf{\theta}{1}(x_1) = \tDist{\theta}1(x_1)/\r{\theta}{1}(x_1)$.
The \smc sampler is summarized in Algorithm~\ref{alg:smc:smc}.

\begin{algorithm}
  \caption{Sequential Monte Carlo (each step is for $i = \range{1}{\Np}$)}
  \label{alg:smc:smc}
  \begin{algorithmic}[1]
    \STATE Draw $x_1^i \sim \r{\theta}{1}(x_1)$. %for $i = \range{1}{\Np}$.
    \STATE Set $w_1^i = \wf{\theta}{1}(x_1^i)$. %for $i = \range{1}{\Np}$.
    \FOR{$t = 2$ \TO $\T$}
    \STATE Draw $\{ a_t^i, x_t^i \} \sim \M{\theta}{t}( a_t, x_t )$. %for  $i =\range{1}{\Np}$.
    \STATE Set $x_{1:t}^i = ( x_{1:t-1}^{a_t^i}, x_t^i )$. %for $i = \range{1}{\Np}$.
    \STATE Set $w_t^i = \wf{\theta}{t}(x_{1:t}^{i})$. %for $i = \range{1}{\Np}$.
    \ENDFOR
  \end{algorithmic}
\end{algorithm}

It is interesting to note that the joint law of all the random variables generated by Algorithm~\ref{alg:smc:smc}
can be written down explicitly. Let
\begin{align*}
\xx_t &= \crange{x_t^1}{x_t^\Np} &&\text{and } & \aa_t &= \crange{a_t^1}{a_t^\Np},
\end{align*}
refer to all the particles and ancestor indices, respectively, generated at time~$t$ of the algorithm.
It follows that the \smc sampler generates a collection of random variables
$\{\xx_{1:\T}, \aa_{2:\T}\} \in \setX^{\Np\T} \times \crange{1}{N}^{\Np(\T-1)}$.
Furthermore, $\{a_t^i, x_t^i\}_{i=1}^\Np$ are drawn independently (conditionally on the particle system
generated up to time $t-1$) from the proposal kernel $\M{\theta}{t}$, and similarly at time $t=1$. Hence, the joint probability density function
(with respect to a natural product of $dx$ and counting measure) of these variables is given by
\begin{align}
  \psi_\theta(\xx_{1:\T}, \aa_{2:\T}) \triangleq \prod_{i = 1}^\Np \r{\theta}{1}(x_1^i) \prod_{t = 2}^{\T} \prod_{i = 1}^\Np \M{\theta}{t}(a_t^i, x_t^i).
\end{align}

%%%%%%%%%%%%%%%%%%%%%%%
%     NEW SECTION     %
%%%%%%%%%%%%%%%%%%%%%%%
\section{The \pgas kernel}\label{sec:pgas}%
We now turn to the construction of \pgas, a family of Markov kernels on the space of trajectories $(\setX^\T, \sigmaX^\T)$.
We will provide an algorithm for generating samples from these Markov kernels, which are
thus defined implicitly by the algorithm.

\subsection{Particle Gibbs}
Before stating the \pgas algorithm, we review the main ideas of the \pg algorithm 
of \cite{AndrieuDH:2010} and we then turn to our proposed modification of 
this algorithm via the introduction of an \emph{ancestor sampling} step.

\pg is based on an \smc sampler, akin to a standard
\pf, but with the difference that one particle trajectory is specified \emph{a priori}.
This path, denoted as $x_{1:\T}^\prime = \prange{x_1^\prime}{x_\T^\prime}$, serves
as a reference trajectory.
Informally, it can be thought of as guiding the simulated particles to a relevant region of the state space.
After a complete pass of the \smc algorithm, a trajectory $x_{1:\T}^\star$ is sampled from among the
particle trajectories. That is, we draw $x_{1:\T}^\star$ with $\Prb(x_{1:\T}^\star = x_{1:\T}^i) \propto w_\T^i$.
This procedure thus maps $x_{1:\T}^\prime$ to a probability distribution on $\sigmaX^\T$, implicitly defining
a Markov kernel on $(\setX^\T, \sigmaX^\T)$.

In a standard \pf, the samples $\{a_t^i, x_t^i\}$ are drawn independently from the proposal kernel \eqref{eq:smc:joint_proposal}
for $i = \range{1}{\Np}$. When sampling from the \pg kernel, however,
we condition on the event that the reference trajectory $x_{1:\T}^\prime$ is retained throughout the sampling procedure.
To accomplish this, we sample according to \eqref{eq:smc:joint_proposal}
only for $i=\range{1}{\Np-1}$. The $\Np$th particle and its ancestor index are then set deterministically as $x_t^\Np = x_t^\prime$ and $a_t^\Np = \Np$.
This implies that after a complete pass of the algorithm, the $\Np$th particle path coincides with the reference trajectory,
\ie, $x_{1:\T}^\Np = x_{1:\T}^\prime$.

The fact that $x_{1:\T}^\prime$ is used as a reference trajectory in the \smc sampler implies an invariance
property of the \pg kernel which is of key relevance. More precisely, as show by \cite[Theorem~5]{AndrieuDH:2010},
for any number of particles $\Np \geq 1$ and for any $\theta \in \Theta$, the PG kernel leaves the exact target distribution $\ntDist{\theta}{\T}$ invariant.
We return to this invariance property below, when it is shown to hold also for the proposed \pgas kernel.

\subsection{Ancestor sampling}
As noted above, the \pg algorithm keeps the reference trajectory $x_{1:\T}^\prime$ intact throughout the
sampling procedure. While this results in a Markov kernel which leaves $\ntDist{\theta}{\T}$
invariant, it has been recognized that the mixing properties of this kernel can be very poor due to path degeneracy \cite{LindstenS:2013,ChopinS:2013}.

To address this fundamental problem we now turn to our new procedure, \pgas.
The idea is to sample a new value for the index variable $a_t^\Np$ in an \emph{ancestor sampling} step.
While this is a small modification of the algorithm, the improvement in mixing can be quite considerable;
see Section~\ref{sec:pgas:degeneracy} and the numerical evaluation in Section~\ref{sec:eval}.
The ancestor sampling step is implemented as follows.

At time $t \geq 2$, we consider the part of the reference trajectory $x_{t:\T}^\prime$ ranging from the current time $t$ to the
final time point $\T$. The task is to artificially assign a history to this partial path. This is done by connecting $x_{t:\T}^\prime$
to one of the particles $\{ x_{1:t-1}^i \}_{i=1}^\Np$. Recall that the ancestry
of a particle is encoded via the corresponding ancestor index. Hence, 
we can connect the partial reference path to one of the particles $\{ x_{1:t-1}^i \}_{i=1}^\Np$ by assigning
a value to the variable $a_t^\Np \in \crange{1}{\Np}$. To do this, first we compute the weights
\begin{align}
  \label{eq:pgas:as_weights}
  \widetilde w_{t-1\mid\T}^i \triangleq w_{t-1}^i \frac{\tDist{\theta}{\T}( (x_{1:t-1}^i, x_{t:\T}^\prime ) )}{ \tDist{\theta}{t-1} (x_{1:t-1}^i) }
\end{align}
for $i = \range{1}{\Np}$. Here, $(x_{1:t-1}^i, x_{t:\T}^\prime )$ refers to the point in $\setX^\T$ formed by concatenating
the two partial trajectories. Then, we sample $a_t^\Np$ with $\Prb(a_t^\Np = i) \propto \widetilde w_{t-1\mid \T}^{i}$.
The expression above can be understood as an application of Bayes' theorem, where the importance weight $w_{t-1}^i$ is
the prior probability of the particle $x_{1:t-1}^i$ and the ratio between the target densities in \eqref{eq:pgas:as_weights}
can be seen as the likelihood that $x_{t:\T}^\prime$ originated from $x_{1:t-1}^i$.
A formal argument for why \eqref{eq:pgas:as_weights} provides the correct ancestor sampling distribution,
in order to retain the invariance properties of the kernel,
is detailed in the proof of Theorem~\ref{thm:pgas:invariance} in Section~\ref{sec:justification}.

The sampling procedure outlined above is summarized in Algorithm~\ref{alg:pgas:pgas} and
the family of \pgas kernels is formally defined below.
Note that the only difference between \pg and \pgas is on line~\ref{row:atN} of Algorithm~\ref{alg:pgas:pgas} (where, for \pg, we
would simply set $a_t^\Np = \Np$). However, as we shall see, the effect of this small modification
on the mixing of the kernel is quite significant.

\begin{definition}[\pgas kernels]\label{def:pgas}
  For any $\Np \geq 1$ and any $\theta \in \Theta$, Algorithm~\ref{alg:pgas:pgas}
  maps $x_{1:\T}^\prime$ stochastically into $x_{1:\T}^\star$, thus implicitly defining a Markov kernel $\kPGAS$ on $(\setX^\T, \sigmaX^\T)$.
  The family of Markov kernels $\{ \kPGAS : \theta \in \Theta \}$, indexed by $\Np \geq 1$, is referred to as the
  \pgas family of kernels.
\end{definition}

\begin{algorithm}
  \caption{\pgas Markov kernel}
  \label{alg:pgas:pgas}
  \begin{algorithmic}[1]
    \REQUIRE Reference trajectory $x_{1:\T}^\prime \in \setX^\T$.
    \STATE \label{row:x1}Draw $x_1^i \sim \r{\theta}{1}(x_1)$ for $i = \range{1}{\Np-1}$. 
    \STATE Set $x_1^{\Np} = x_1^\prime$.
    \STATE Set $w_1^i = \wf{\theta}{1}(x_1^i)$ for $i = \range{1}{\Np}$.
    \FOR{$t = 2$ \TO $\T$}
    \STATE \label{row:atxt} Draw $\{ a_t^i, x_t^i \} \sim \M{\theta}{t}( a_t, x_t )$ for  $i =\range{1}{\Np-1}$.
    \STATE Set $x_t^{\Np} = x_{t}^\prime$.
    \STATE Compute $\{ \widetilde w_{t-1\mid\T}^i \}_{i=1}^\Np$ according to \eqref{eq:pgas:as_weights}.
    \STATE \label{row:atN} Draw $a_t^{\Np}$ with $\Prb(a_t^{\Np} = i) \propto \widetilde w_{t-1\mid\T}^i$.
    \STATE Set $x_{1:t}^i = ( x_{1:t-1}^{a_t^i}, x_t^i )$ for $i = \range{1}{\Np}$.
    \STATE Set $w_t^i = \wf{\theta}{t}(x_{1:t}^{i})$ for $i = \range{1}{\Np}$.
    \ENDFOR
    \STATE \label{row:k} Draw $k$ with $\Prb(k = i) \propto w_\T^i$.
    \RETURN $x_{1:\T}^\star = x_{1:\T}^k$.
  \end{algorithmic}
\end{algorithm}

\subsection{The effect of path degeneracy on \pg and on \pgas}\label{sec:pgas:degeneracy}
We have argued that ancestor sampling can considerably improve the mixing of \pg.
To illustrate this effect and to provide an explanation of its cause, we consider a simple numerical example.
Further empirical evaluation of \pgas is provided in Section~\ref{sec:eval}.
Consider the stochastic volatility model,
\begin{subequations}
  \label{eq:pgas:sv_model}
  \begin{align}
    x_{t+1} &= ax_t + v_t, & v_t &\sim \N(0, \sigma^2), \\
    \label{eq:pgas:sv_model_b}
    y_t &= e_t \exp\left(\half x_t \right), & e_t &\sim \N(0,1),
  \end{align}
\end{subequations}
where the state process $\process{x_t}$ is latent and observations are made only via
the measurement process $\process{y_t}$. Similar models have been used to generalize
the Black-Scholes option pricing equation to allow for the variance to change over time \cite{ChesneyS:1989,MelinoT:1990}.

For simplicity, the parameter $\theta = (a, \sigma) = (0.9, 0.5)$ is assumed to be known.
A batch of $\T = 400$ observations are simulated from the system. Given these, we
seek the joint smoothing density $p(x_{1:\T} \mid y_{1:\T})$. To generate samples
from this density we employ both \pg and \pgas with varying number of particles ranging
from $\Np = 5$ to $\Np = \thsnd{1}$. We simulate sample paths of length $\thsnd{1}$ for each algorithm.
To compare the mixing, we look at the update rate of $x_t$ versus $t$, which is defined
as the proportion of iterations where $x_t$ changes value. The results
are reported in Figure~\ref{fig:pgas:degeneracy_illustration}, which reveals that
ancestor sampling significantly increases the probability of updating $x_t$ for $t$ far from $\T$.

\begin{figure}[ptb]
  \centering
  \includegraphics[width=0.5\columnwidth]{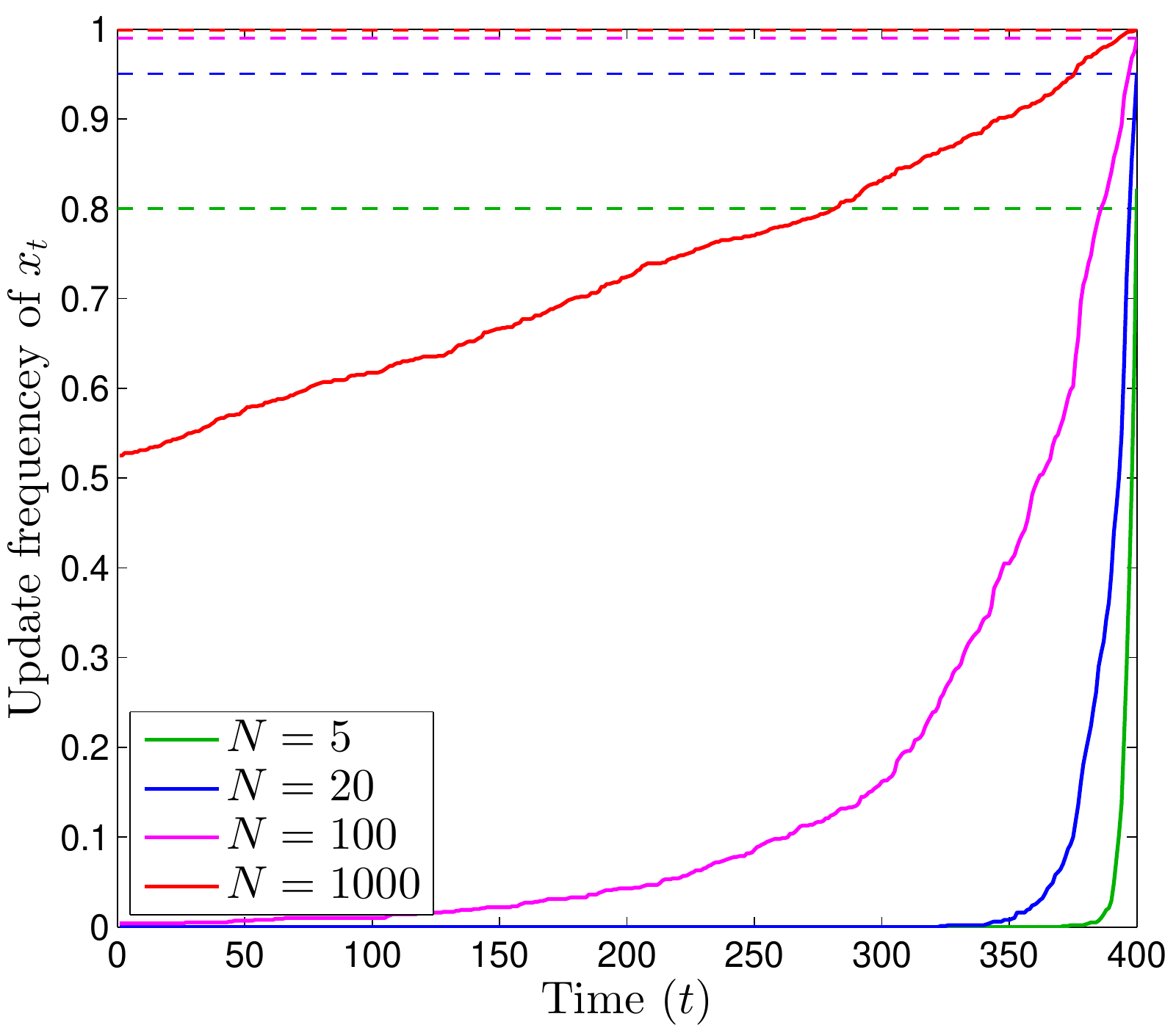}%
  \includegraphics[width=0.5\columnwidth]{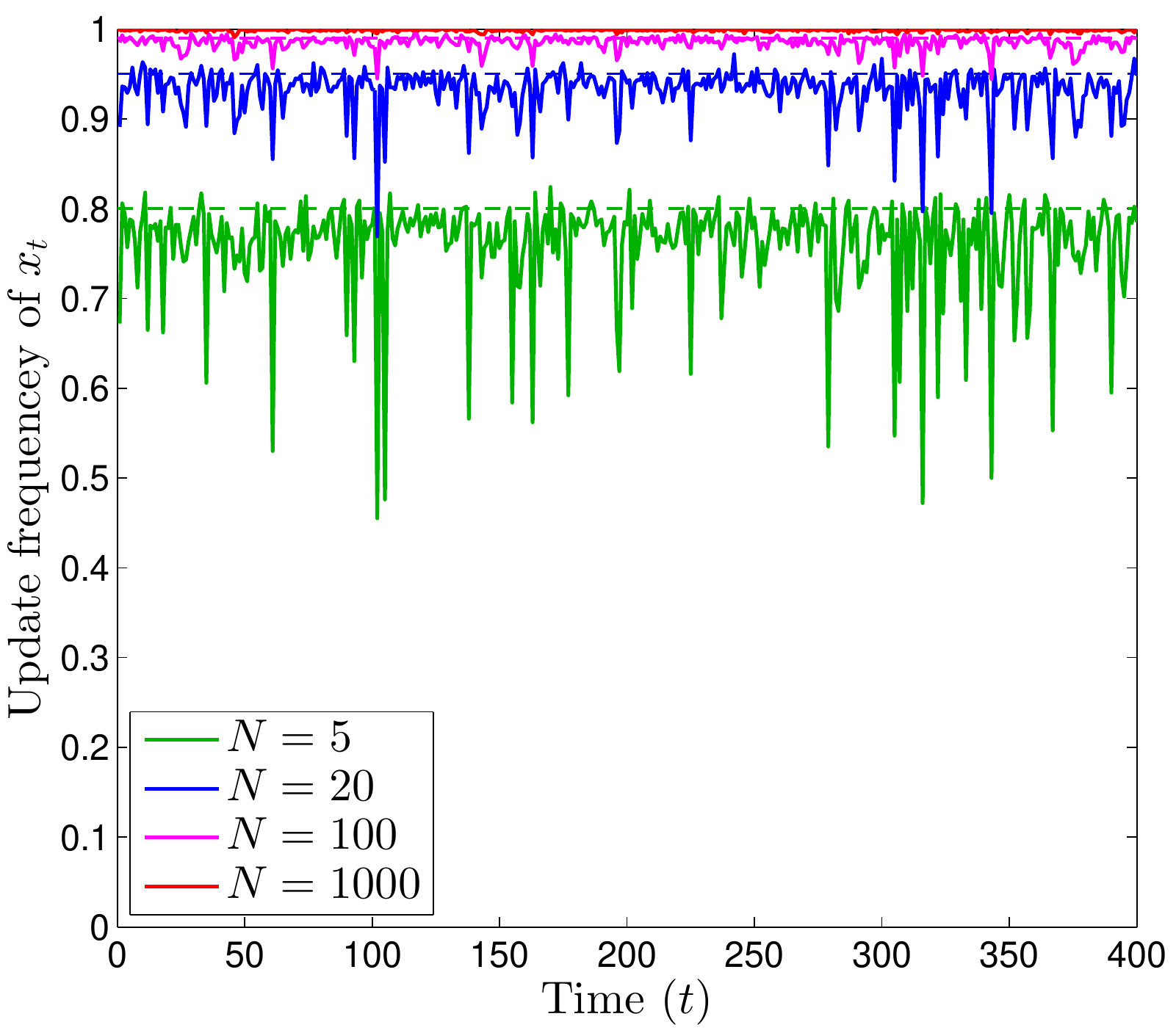}
  \caption{Update rates for $x_t$ versus $t \in \crange{1}{400}$ for \pg (left) and for \pgas (right). The dashed lines correspond to the ideal
  rates $(\Np-1)/\Np$.     (This figure is best viewed in color.)}
  \label{fig:pgas:degeneracy_illustration}
\end{figure}

The poor update rates for \pg is a manifestation of the well known path degeneracy problem of \smc samplers (see, \eg, \cite{DoucetJ:2011}).
Consider the process of sampling from the PG kernel for a fixed reference trajectory $x_{1:\T}^\prime$.
A particle system generated by the PG algorithm (corresponding to Algorithm~\ref{alg:pgas:pgas}, but with line~\ref{row:atN}
replaced with $a_t^\Np = \Np$) is shown in Figure~\ref{fig:pgas:pgex_consecutive} (left).
For clarity of illustration, we have used a small number of particles and time steps, $\Np = 20$ and $\T = 50$, respectively.
By construction the reference trajectory (shown by a thick blue line) is
retained throughout the sampling procedure. As a consequence, the particle system degenerates toward this trajectory
which implies that $x_{1:\T}^\star$ (shown as a red line) to a large extent will be identical to $x_{1:\T}^\prime$.

What is, perhaps, more surprising is that \pgas is so much more insensitive to the degeneracy issue.
To understand why this is the case, we analyze the procedure for sampling from
the PGAS kernel $\kPGAS(x_{1:\T}^\prime, \cdot)$ for the same reference trajectory $x_{1:\T}^\prime$ as above.
The particle system generated by Algorithm~\ref{alg:pgas:pgas} (with ancestor sampling)
is shown in Figure~\ref{fig:pgas:pgex_consecutive} (right). The thick blue lines are again
used to illustrate the reference particles, but now with updated ancestor indices. That is,
the blue line segments are drawn between $x_{t-1}^{a_t^\Np}$ and $x_t^\prime$ for $t \geq 2$.
It can be seen that the effect of ancestor sampling is that, informally, the reference trajectory is broken into pieces.
It is worth pointing out that the particle system still collapses; ancestor sampling does not prevent path degeneracy.
However, it causes the particle system to degenerate toward something
different than the reference trajectory. As a consequence, $x_{1:\T}^\star$ (shown as a red line in the figure) will with high probability be
substantially different from $x_{1:\T}^\prime$, enabling high update rates and thereby much faster mixing.

\begin{figure}[ptb]
  \includegraphics[width=0.5\columnwidth]{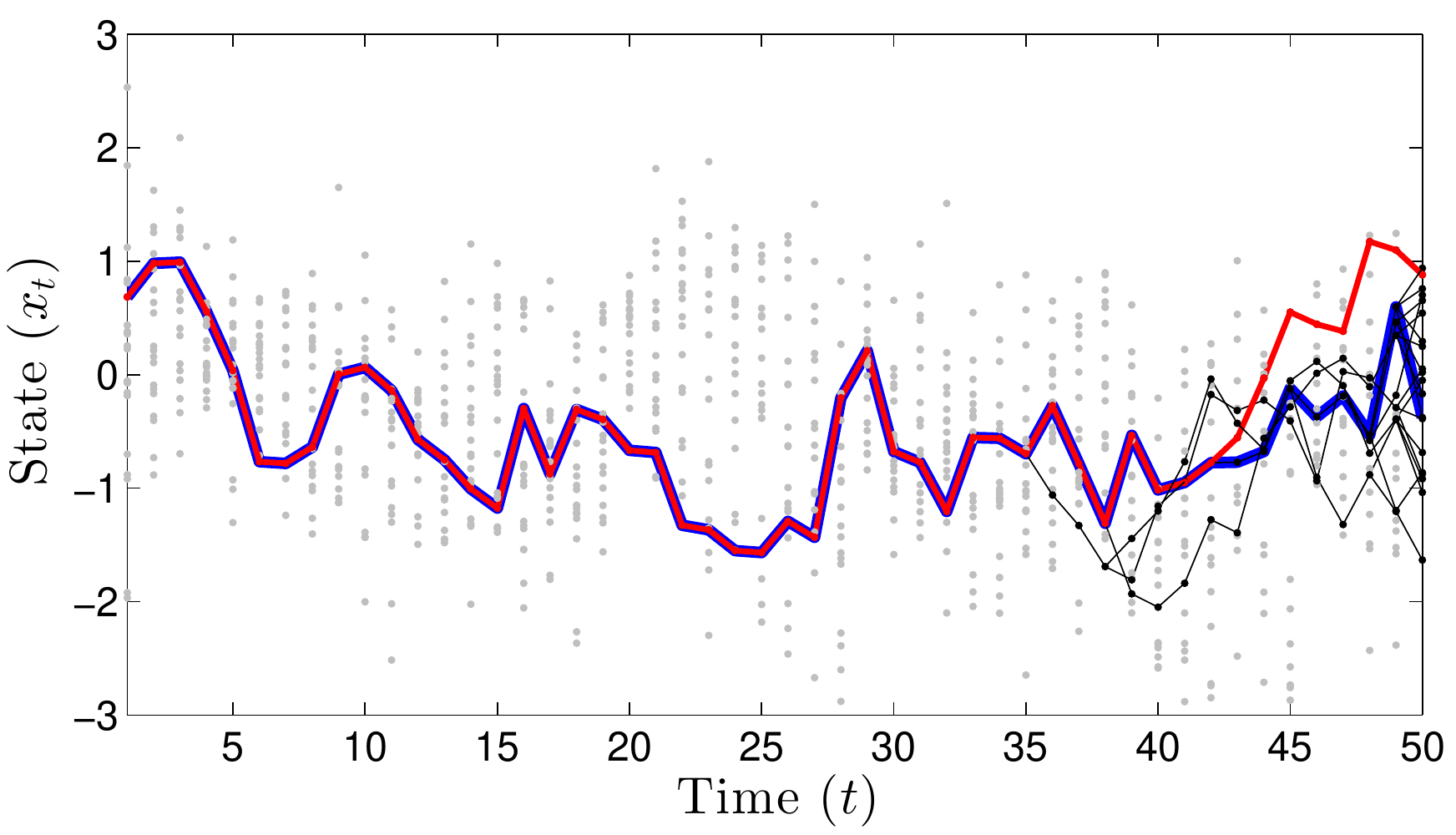}%
  \includegraphics[width=0.5\columnwidth]{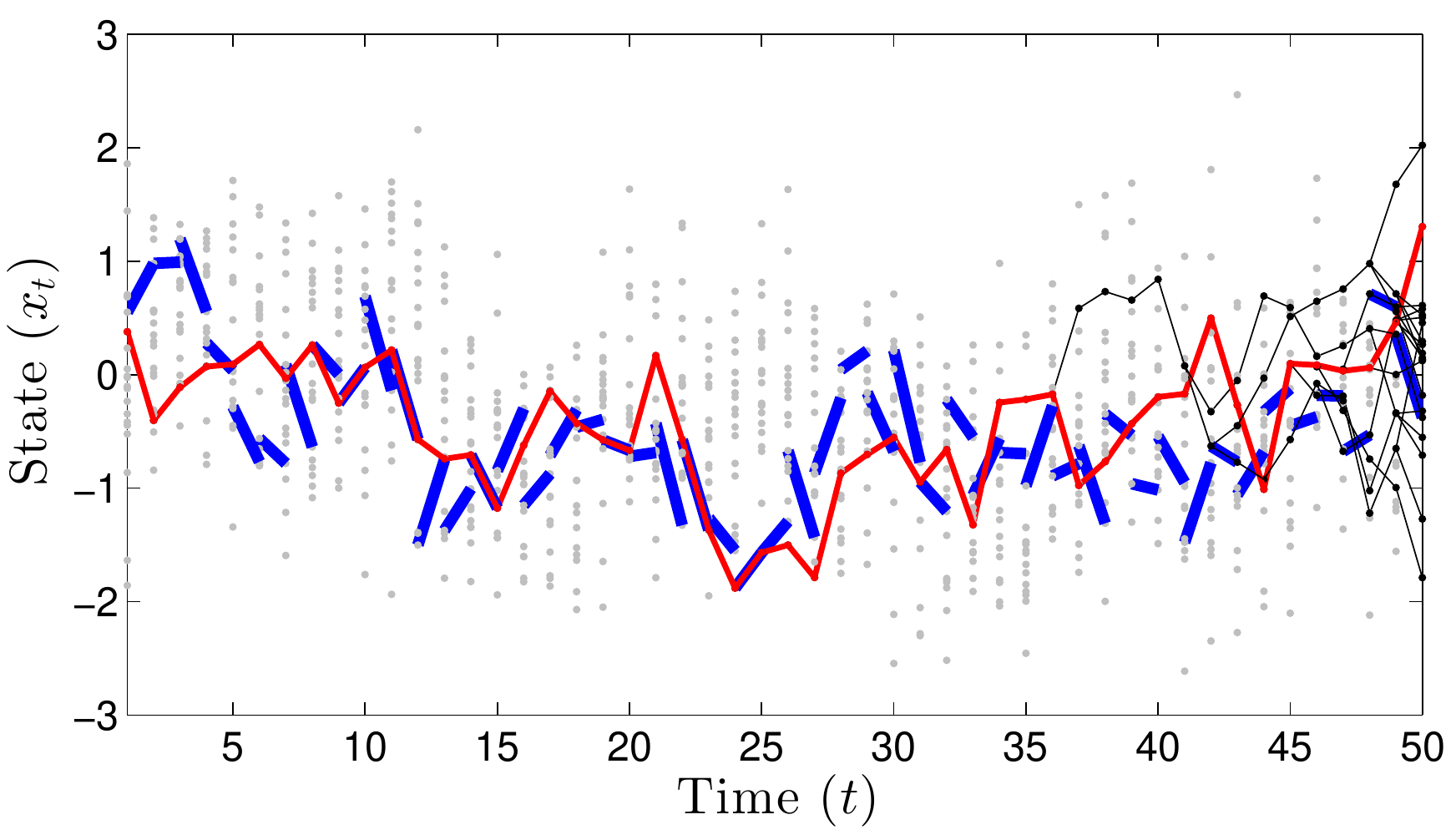}
 \caption{Particle systems generated by the \pg algorithm (left) and by the \pgas algorithm (right), for the same reference trajectory $x_{1:\T}^\prime$ (shown as a thick blue line in the left panel).
    The gray dots show the particle positions and the thin black lines show the ancestral dependencies of the particles. The extracted trajectory
    $x_{1:\T}^\star$ is illustrated with a red line. Note that, due to path degeneracy, the particles shown as grey dots are not
    reachable by tracing any of the ancestral lineages from time $\T$ and back. In the right panel, ancestor sampling has the effect of
    breaking the reference trajectory into pieces, causing the particle system to degenerate toward something different than $x_{1:\T}^\prime$.
    (This figure is best viewed in color.)}
  \label{fig:pgas:pgex_consecutive}
\end{figure}

%%%%%%%%%%%%%%%%%%%%%%%%%%%%%%%%%
%%%%%%%%%% SECTION %%%%%%%%%%%%%%
%%%%%%%%%%%%%%%%%%%%%%%%%%%%%%%%%
\section{Theoretical justification}\label{sec:justification}%
\subsection{Stationary distribution}
We begin by stating a theorem, whose proof is provided later in this section,
which shows that the invariance property of \pg is not violated
by the ancestor sampling step.

\begin{theorem}
  \label{thm:pgas:invariance} For any $\Np \geq 1$ and $\theta \in \Theta$, the \pgas kernel $\kPGAS$ leaves $\ntDist{\theta}{\T}$ invariant:
  \begin{align*}
    \ntDist{\theta}{\T}(B) &= \int \kPGAS(x_{1:\T}^\prime, B) \ntDist{\theta}{\T}(dx_{1:\T}^\prime), & \forall B&\in \sigmaX^\T.
  \end{align*}
\end{theorem}

An apparent difficulty in establishing this result is that it is not possible to write down a simple,
closed-form expression for $\kPGAS$. In fact, the \pgas kernel is given by
\begin{align}
  \label{eq:pgas:kernel_as_expectation}
  \kPGAS(x_{1:\T}^\prime, B) = \E_{\theta,x_{1:\T}^\prime} \left[ \I_B(x_{1:\T}^k ) \right],
\end{align}
where $\I_B$ is the indicator function for the set $B$ and where $\E_{\theta,x_{1:\T}^\prime}$ denotes expectation with respect 
to all the random variables generated by Algorithm~\ref{alg:pgas:pgas}, \ie, all the particles  $\xx_{1:\T}$ and ancestor indices $\aa_{2:\T}$,
 as well as the index $k$.
Computing this expectation is not possible in general.
Instead of working directly with \eqref{eq:pgas:kernel_as_expectation}, however, we 
can adopt a strategy employed by \cite{AndrieuDH:2010}, treating all the random 
variables generated by Algorithm~\ref{alg:pgas:pgas}, $\{\xx_{1:\T}, \aa_{2:\T}, k \}$, 
as auxiliary variables, thus avoiding an intractable integration.
In the following, it is convenient to view $x_t^\Np$ as a random variable with distribution $\delta_{x_t^\prime}$.

Recall that the particle trajectory $x_{1:\T}^{k}$ is the ancestral path of the particle $x_\T^k$. That is, we can write
\begin{align}
  \label{eq:pgas:trajectorydef}
  x_{1:\T}^k = x_{1:\T}^{b_{1:\T}} \triangleq \prange{x_{1}^{b_1}}{x_\T^{b_\T}},
\end{align}
where the indices $b_{1:\T}$ are given recursively by the ancestor indices: $b_\T = k$ and $b_t = a_{t+1}^{b_{t+1}}$.
Let $\extendedspace \triangleq \setX^{\Np\T}\times\crange{1}{\Np}^{\Np(\T-1)+1}$ be the space of all random variables generated by Algorithm~\ref{alg:pgas:pgas}.
Following \cite{AndrieuDH:2010}, we then define a function $\phi_\theta: \extendedspace\mapsto\reals$ as follows:
\begin{align}
  \nonumber
  \phi_\theta(\xx_{1:\T}, \aa_{2:\T}, k)
  &= \phi_\theta(x_{1:\T}^{b_{1:\T}}, b_{1:\T})
  \phi_\theta(\xx_{1:\T}^{-b_{1:\T}}, \aa_{2:\T}^{-b_{2:\T}} \mid x_{1:\T}^{b_{1:\T}}, b_{1:\T})\\
  \label{eq:pgas:phidef}
  &\triangleq \underbrace{\vphantom{\prod_{\substack{i=1 \\ i\neq b_1 }}^\Np} % For vertical alignment
      \frac{\ntDist{\theta}{\T}( x_{1:\T}^{b_{1:\T}})}{\Np^\T}}_{\text{marginal}}
    \underbrace{\prod_{\substack{i=1 \\ i\neq b_1 }}^\Np \r{\theta}{1}(x_1^i)\prod_{t = 2}^{\T} \prod_{\substack{i=1 \\ i\neq b_t }}^\Np \M{\theta}{t}(a_t^i, x_t^i)
    }_{\text{conditional}},
\end{align}
where we have introduced the notation
\begin{align*}
\xx_t^{-i} &= \{\range{x_t^1}{x_t^{i-1},\,\range{x_t^{i+1}}{x_t^\Np}}\},&
\xx_{1:\T}^{-b_{1:\T}} &= \crange{\xx_1^{-b_1}}{\xx_\T^{-b_\T}}
\end{align*}
and similarly for the ancestor indices.
By construction, $\phi_\theta$ is nonnegative and integrates to one, \ie, $\phi_\theta$ is a probability density function on $\Omega$.
 We refer to this density as the \emph{extended target density}.

The factorization into a marginal and a conditional density
is intended to reveal some of the structure inherent in the extended target density. In particular, the marginal density of the variables
$\{ x_{1:\T}^{b_{1:\T}}, b_{1:\T} \}$ is defined to be equal to the original target density
$\ntDist{\theta}{\T}( x_{1:\T}^{b_{1:\T}})$, up to a factor $\Np^{-\T}$ related to the index variables $b_{1:\T}$.
This has the important implication that if $\{ \xx_{1:\T}, \aa_{2:\T}, k \}$ are distributed according to $\phi_\theta$, then,
by construction, the marginal distribution of $x_{1:\T}^{b_{1:\T}}$ is $\ntDist{\theta}{\T}$.

By constructing an \mcmc kernel with invariant distribution $\phi_\theta$, we will thus
obtain a kernel with invariant distribution $\ntDist{\theta}{\T}$ (the \pgas kernel) as a byproduct.
To prove Theorem~\ref{thm:pgas:invariance} we will reinterpret all the steps of the \pgas algorithm as partially collapsed Gibbs steps for $\phi_\theta$.
The meaning of partial collapsing will be made precise in the proof of Lemma~\ref{lem:pgas:partialcollapse} below,
but basically it refers to the process of marginalizing out some of the variables of the model in the individual steps of the Gibbs sampler.
This is done in such a way that it does not violate the invariance property of the Gibbs kernel, 
\ie, each such Gibbs step will leave the extended target distribution invariant.
As a consequence, the invariance property of the \pgas kernel follows.
First we show that the \pgas algorithm in fact implements the following sequence of partially collapsed Gibbs steps for~$\phi_\theta$.

\begin{proc}[Instrumental reformulation of \pgas]
  \label{proc:pgas:pgas_instrumental}
  Given $x_{1:\T}^{\prime,b_{1:\T}^\prime} \in \setX^\T$  and $b_{1:\T}^\prime \in \crange{1}{\Np}^\T$:
  \begin{enumerate}
  \item[\gibbsnum{i}] Draw $\xx_{1}^{-b_{1}^\prime} \sim \phi_\theta( \,\cdot \mid  x_{1:\T}^{\prime,b_{1:\T}^\prime}, b_{1:\T}^\prime)$ and,
    for $t = 2$ to $\T$, draw:
    \begin{align*}
      \{\xx_t^{-b_t}, \aa_t^{ -b_t}\} \sim {}& \phi_\theta( \,\cdot \mid \xx_{1:t-1}^{-b_{1:t-1}^\prime}, \aa_{2:t-1}, x_{1:\T}^{\prime,b_{1:\T}^\prime}, b_{t-1:\T}^\prime), \\
      a_{t}^{b_{t}} \sim {}&\phi_\theta( \,\cdot \mid \xx_{1:t-1}^{-b_{1:t-1}^\prime}, \aa_{2:t-1}, x_{1:\T}^{\prime,b_{1:\T}^\prime}, b_{t:\T}^\prime),
    \end{align*}
  \item[\gibbsnum{ii}] Draw $k \sim {}\phi_\theta( \,\cdot \mid \xx_{1:\T}^{-b_{1:\T}}, \aa_{2:\T}, x_{1:\T}^{\prime,b_{1:\T}^\prime})$.
  \end{enumerate}
\end{proc}

\begin{lemma}\label{lem:pgas:invariance}%
  Algorithm~\ref{alg:pgas:pgas} is equivalent to the partially collapsed Gibbs sampler of Procedure~\ref{proc:pgas:pgas_instrumental},
  conditionally on $x_{1:\T}^{\prime,b_{1:\T}^\prime} = x_{1:\T}^\prime$ and $b_{1:\T}^\prime = \prange{\Np}{\Np}$.
\end{lemma}
\begin{proof}
  From \eqref{eq:pgas:phidef} we have, by construction,
  \begin{align}
    \phi_\theta(\xx_{1:\T}^{-b_{1:\T}}, \aa_{2:\T}^{-b_{2:\T}} \mid x_{1:\T}^{b_{1:\T}}, b_{1:\T}) =
    \prod_{\substack{i=1 \\ i\neq b_1 }}^\Np \r{\theta}{1}(x_1^i)\prod_{t = 2}^{\T} \prod_{\substack{i=1 \\ i\neq b_t }}^\Np \M{\theta}{t}(a_t^i, x_t^i).
  \end{align}
  By marginalizing this expression  over $\{ \xx_{t+1:\T}^{-b_{t+1:\T}}, \aa_{t+1:\T}^{-b_{t+1:\T}} \}$ we get
  \begin{align}
    \label{eq:pgad:marginalconditional}
    \phi_\theta(\xx_{1:t}^{-b_{1:t}}, \aa_{2:t}^{-b_{2:t}} \mid x_{1:\T}^{b_{1:\T}}, b_{1:\T})
    =\prod_{\substack{i=1 \\ i\neq b_1 }}^\Np \r{\theta}{1}(x_1^i)\prod_{s = 2}^{t} \prod_{\substack{i=1 \\ i\neq b_s }}^\Np \M{\theta}{s}(a_s^i, x_s^i),
  \end{align}
  It follows that
  \begin{subequations}
    \begin{align}
      \label{eq:pgas:cpf_reinterpreted_a}
      \phi_\theta(\xx_{1}^{-b_{1}} \mid x_{1:\T}^{b_{1:\T}}, b_{1:\T}) &= \prod_{\substack{i=1 \\ i\neq b_1 }}^\Np \r{\theta}{1}(x_1^i),
    \end{align}
    and, for $t = \range{2}{\T}$,
    \begin{align}
      \nonumber
      \phi_\theta{}&(\xx_{t}^{-b_{t}}, \aa_{t}^{-b_{t}} \mid \xx_{1:t-1}^{-b_{1:t-1}}, \aa_{2:t-1}^{-b_{2:t-1}}, x_{1:\T}^{b_{1:\T}}, b_{1:\T}) \\
      \label{eq:pgas:cpf_reinterpreted_b}
      &= \frac{ \phi_\theta(\xx_{1:t}^{-b_{1:t}}, \aa_{2:t}^{-b_{2:t}} \mid x_{1:\T}^{b_{1:\T}}, b_{1:\T})}
      {\phi_\theta(\xx_{1:t-1}^{-b_{1:t-1}}, \aa_{2:t-1}^{-b_{2:t-1}} \mid x_{1:\T}^{b_{1:\T}}, b_{1:\T})}
      = \prod_{\substack{i=1 \\ i\neq b_t }}^\Np \M{\theta}{t}(a_t^i, x_t^i).
    \end{align}
  \end{subequations}
  Hence, we can sample from \eqref{eq:pgas:cpf_reinterpreted_a} and \eqref{eq:pgas:cpf_reinterpreted_b} by
  drawing $x_{1}^i \sim \r{\theta}{1}(\cdot)$ for $i \in \crange{1}{\Np}\setminus b_1$ and
  $\{a_t^i, x_{t}^i\} \sim \M{\theta}{t}(\cdot)$ for $i \in \crange{1}{\Np}\setminus b_t$, respectively.
  Consequently, with the choice $b_{t} = \Np$ for $t=\range{1}{\T}$, the initialization at line~\ref{row:x1} and the particle propagation at line~\ref{row:atxt}
  of Algorithm~\ref{alg:pgas:pgas} correspond to sampling from \eqref{eq:pgas:cpf_reinterpreted_a} and \eqref{eq:pgas:cpf_reinterpreted_b}, respectively.

  Next, we consider the ancestor sampling step. Recall that $a_t^{b_t}$ identifies to $b_{t-1}$. We can thus write
  \begin{align}
    \nonumber
    \phi_\theta(a_t^{b_t} &{}\mid \xx_{1:t-1}, \aa_{2:t-1}, x_{t:\T}^{b_{t:\T}}, b_{t:\T})
    \propto \phi_\theta(\xx_{1:t-1}, \aa_{2:t-1}, x_{t:\T}^{b_{t:\T}}, b_{t-1:\T}) \\
    \nonumber
    &= \phi_\theta(x_{1:\T}^{b_{1:\T}}, b_{1:\T})
    \phi_\theta(\xx_{1:t-1}^{-b_{1:t-1}}, \aa_{2:t-1}^{-b_{2:t-1}} \mid x_{1:\T}^{b_{1:\T}}, b_{1:\T}) \\
    \label{eq:pgas:as_gibbsstep_expression}
    &=  \frac{\tDist{\theta}{\T}(x_{1:\T}^{b_{1:\T}})}{\tDist{\theta}{t-1}(x_{1:t-1}^{b_{1:t-1}})} 
    \frac{\tDist{\theta}{t-1}(x_{1:t-1}^{b_{1:t-1}}) }{ \normconst{\theta}{\T} \Np^\T} \prod_{\substack{i=1 \\ i\neq b_1 }}^\Np
    \r{\theta}{1}(x_1^i)\prod_{s = 2}^{t-1} \prod_{\substack{i=1 \\ i\neq b_s }}^\Np \M{\theta}{s}(a_s^i, x_s^i).
  \end{align}
  To simplify this expression, note first that we can write
  \begin{align}
    \tDist{\theta}{t-1}(x_{1:t-1}) = \tDist{\theta}{1}(x_{1}) \prod_{s=2}^{t-1} \frac{ \tDist{\theta}{s}(x_{1:s})  }{ \tDist{\theta}{s-1}(x_{1:s-1}) }.
  \end{align}
  By using the definition of the weight function \eqref{eq:smc:weightfunction}, this expression can be expanded according to
  \begin{align}
    \tDist{\theta}{t-1}(x_{1:t-1}) = \wf{\theta}{1}(x_1) \r{\theta}{1}(x_1) \prod_{s=2}^{t-1} \wf{\theta}{s} (x_{1:s}) \r{\theta}{s}(x_s \mid x_{1:s-1}).
  \end{align}
  Plugging the trajectory $x_{1:t-1}^{b_{1:t-1}}$ into the above expression, we get
  \begin{align}
    \nonumber
    \tDist{\theta}{t-1}(x_{1:t-1}^{b_{1:t-1}})
    &= w_1^{b_1} \r{\theta}{1} (x_1^{b_1}) \prod_{s = 2}^{t-1} w_{s}^{b_{s}} \r{\theta}{s} (x_s^{b_s} \mid x_{1:s-1}^{b_{1:s-1}} ) \\
    \nonumber
    &= \left( \prod_{s = 1}^{t-1} \sum_{l=1}^\Np w_s^l \right) \frac{w_1^{b_1}}{\normsum{l} w_1^l} \r{\theta}{1}(x_1^{b_1})
    \prod_{s = 2}^{t-1} \frac{w_{s}^{b_{s}}}{\normsum{l} w_s^l} \r{\theta}{s}(x_s^{b_s} \mid x_{1:s-1}^{b_{1:s-1}}) \\
    % &= \left( \prod_{s = 1}^{t-1} \sum_l w_s^l \right) \w_1^{b_1} \r{\theta}{1}(x_1^{b_1}) \prod_{s = 2}^{t-1} \w_{s}^{b_{s}} \r{\theta}{s}(x_s^{b_s} \mid x_{1:s-1}^{b_{1:s-1}}) \\
    \label{eq:pgas:target_factorization}
    &= \frac{w_{t-1}^{b_{t-1}}}{\sum_{l} w_{t-1}^l} \left( \prod_{s = 1}^{t-1} \sum_{l=1}^\Np w_s^l \right) \r{\theta}{1}(x_1^{b_1}) \prod_{s = 2}^{t-1} \M{\theta}{s}(a_s^{b_s}, x_s^{b_s}).
    \end{align}
    Expanding the numerator in \eqref{eq:pgas:as_gibbsstep_expression} according to \eqref{eq:pgas:target_factorization} results in
    \begin{align}
      \nonumber
      \phi_\theta( &{}a_t^{b_t}\mid \xx_{1:t-1}, \aa_{2:t-1}, x_{t:\T}^{b_{t:\T}}, b_{t:\T}) \\
      \nonumber
      &\propto
      \frac{\tDist{\theta}{\T}(x_{1:\T}^{b_{1:\T}})}{\tDist{\theta}{t-1}(x_{1:t-1}^{b_{1:t-1}})}
      \frac{w_{t-1}^{b_{t-1}}}{\sum_{l} w_{t-1}^l}
    \frac{  \left( \prod_{s = 1}^{t-1} \sum_{l} w_s^l \right)  }{ \normconst{\theta}{\T} \Np^\T}
    \prod_{i=1}^\Np \r{\theta}{1}(x_1^i)\prod_{s = 2}^{t-1} \prod_{i=1}^\Np \M{\theta}{s}(a_s^i, x_s^i) \\
    \label{eq:pgas:cpf_reinterpreted_as}
    &\propto
    w_{t-1}^{b_{t-1}} \frac{\tDist{\theta}{\T}( (x_{1:t-1}^{b_{1:t-1}}, x_{t:\T}^{b_{t:\T}}) )}{\tDist{\theta}{t-1}(x_{1:t-1}^{b_{1:t-1}})}.
    \end{align}
    Consequently, with $b_t = \Np$ and $x_{t:\T}^{b_{t:\T}} = x_{t:\T}^\prime$, sampling from \eqref{eq:pgas:cpf_reinterpreted_as}
    corresponds to the ancestor sampling step of line \ref{row:atN} of Algorithm~\ref{alg:pgas:pgas}.
    Finally, analogously to \eqref{eq:pgas:cpf_reinterpreted_as}, it follows that
    $\phi_\theta( k \mid \xx_{1:\T}, \aa_{2:\T}) \propto w_\T^k$, which corresponds to line~\ref{row:k} of Algorithm~\ref{alg:pgas:pgas}.
\end{proof}

Next, we show that Procedure~\ref{proc:pgas:pgas_instrumental} leaves $\phi_\theta$ invariant. This is done by concluding that
the procedure is a properly collapsed Gibbs sampler; see \cite{DykP:2008}.
Marginalization, or collapsing, is commonly used within Gibbs sampling to improve the mixing and/or to simplify the sampling procedure.
However, it is crucial that the collapsing is carried out in the correct order to respect the dependencies between the variables of the model.

\begin{lemma}\label{lem:pgas:partialcollapse}
  The Gibbs sampler of Procedure~\ref{proc:pgas:pgas_instrumental} is properly collapsed and thus leaves $\phi_\theta$ invariant.
\end{lemma}
\begin{proof}
  Consider the following sequence of \emph{complete} Gibbs steps:
  \begin{enumerate}
  \item[\gibbsnum{i}] Draw $\{ \xx_{1}^{-b_{1}^\prime}, \underline{\xx}_{2:\T}^{-b_{2:\T}^\prime}, \underline{\aa}_{2:\T}^{-b_{2:\T}^\prime} \} \sim \phi_\theta( \,\cdot \mid  x_{1:\T}^{\prime,b_{1:\T}^\prime}, b_{1:\T}^\prime)$ and,
    for $t = 2$ to $\T$, draw:%
    \begin{align*}
      \{\xx_t^{-b_t}, \aa_t, \underline{\xx}_{t+1:\T}^{-b_{t+1:\T}^\prime}, \underline{\aa}_{t+1:\T}^{-b_{t+1:\T}^\prime}  \} \sim
      \phi_\theta( \,\cdot \mid \xx_{1:t-1}^{-b_{1:t-1}^\prime}, \aa_{2:t-1}, x_{1:\T}^{\prime,b_{1:\T}^\prime}, b_{t:\T}^\prime).      
    \end{align*}
  \item[\gibbsnum{ii}] Draw $k \sim {}\phi_\theta( \,\cdot \mid \xx_{1:\T}^{-b_{1:\T}^\prime}, \aa_{2:\T}, x_{1:\T}^{\prime,b_{1:\T}^\prime})$.
  \end{enumerate}
  In the above, all the samples are drawn from conditionals under the full joint density $\phi_\theta(\xx_{1:\T}, \aa_{2:\T}, k)$.
  Hence, it is clear that the above procedure will leave $\phi_\theta$ invariant.
  Note that some of the variables above have been marked by an underline. It can be seen that these variables are in fact
  never conditioned upon in any subsequent step
  of the procedure. That is, the underlined variables are never used. Therefore, to obtain a valid sampler
  it is sufficient to sample all the non-underlined variables
  from their respective marginals. Furthermore, from \eqref{eq:pgas:cpf_reinterpreted_b} it can be seen that
  $\{ \xx_t^{-b_t}, \aa_t^{-b_t} \}$ are conditionally independent of $a_t^{b_t}$, \ie, it follows that the complete Gibbs
  sweep above is equivalent to the partially collapsed
  Gibbs sweep of Procedure~\ref{proc:pgas:pgas_instrumental}. Hence, the Gibbs sampler is properly collapsed and it will therefore
  leave $\phi_\theta$ invariant.
\end{proof}

\begin{proof}[Proof (Theorem~\ref{thm:pgas:invariance})]
  Let
  \begin{align}
    \mathcal{L}(d\xx_{1:\T}^{-b_{1:\T}^\prime}, d\aa_{2:\T}, dk \mid x_{1:\T}^{\prime}, b_{1:\T}^\prime)
  \end{align}
  denote the law of the random variables generated by Procedure~\ref{proc:pgas:pgas_instrumental},
  conditionally on $x_{1:\T}^{\prime, b_{1:\T}^\prime} = x_{1:\T}^\prime$ and on $b_{1:\T}^\prime$. 
  Using Lemma~\ref{lem:pgas:partialcollapse} and recalling
  that $\phi_\theta(x_{1:\T}^{b_{1:\T}}, b_{1:\T}) = \Np^{-\T} \ntDist{\theta}{\T}(x_{1:\T}^{b_{1:\T}})$ we have
  \begin{align}
    \nonumber
    \ntDist{\theta}{\T}(B) &= \int \I_B(x_{1:\T}^k) \mathcal{L}(d\xx_{1:\T}^{-b_{1:\T}^\prime}, d\aa_{2:\T}, dk \mid x_{1:\T}^{\prime}, b_{1:\T}^\prime) \\
    \label{eq:pgas:pfthm1_invariance1}
    &\hspace{2em}\times \delta_{x_1^\prime}(dx_1^{b_{1}^\prime}) \cdots \delta_{x_\T^\prime}(dx_\T^{b_{\T}^\prime}) \frac{\ntDist{\theta}{\T}(x_{1:\T}^\prime)}{\Np^\T}
    dx_{1:\T}^\prime db_{1:\T}^\prime,& \forall B&\in \sigmaX^\T.
  \end{align}
  % -----------
  By Lemma~\ref{lem:pgas:invariance} we know that Algorithm~\ref{alg:pgas:pgas}, which implicitly defines $\kPGAS$,
  is equivalent to Procedure~\ref{proc:pgas:pgas_instrumental}
  conditionally on $x_{1:\T}^{\prime, b_{1:\T}^\prime} = x_{1:\T}^\prime$ and $b_{1:\T}^\prime = \prange{\Np}{\Np}$.
  That is to say,
  \begin{align}
    \nonumber
    \kPGAS(x_{1:\T}^\prime, B) &= \int \I_B(x_{1:\T}^k) \mathcal{L}(d\xx_{1:\T}^{-\prange{\Np}{\Np}}, d\aa_{2:\T}, dk \mid x_{1:\T}^{\prime}, \prange{\Np}{\Np}) \\
    &\hspace{2em}\times \delta_{x_1^\prime}(dx_1^{\Np}) \cdots \delta_{x_\T^\prime}(dx_\T^{\Np}),
    \hspace{5em}\forall x_{1:\T}^\prime \in \setX^\T, \forall B\in \sigmaX^\T.
  \end{align}
  However, %by considering the steps of Algorithm~\ref{alg:pgas:pgas}, it can be realized that
  the law of $x_{1:\T}^\star$ in Algorithm~\ref{alg:pgas:pgas} is invariant to permutations of the particle indices. That is,
  it does not matter if we place the reference particles on the $\Np$th positions, or on some other positions, when enumerating
  the particles\footnote{A formal proof of this statement is given for the \pg sampler in \cite{ChopinS:2013}.
    The same argument can be used also for \pgas.}. This implies that for any $b_{1:\T}^\prime \in \crange{1}{\Np}^\T$,
  \begin{align}
    \nonumber
    \kPGAS(x_{1:\T}^\prime, B) &= \int \I_B(x_{1:\T}^k) \mathcal{L}(d\xx_{1:\T}^{-b_{1:\T}^\prime}, d\aa_{2:\T}, dk \mid x_{1:\T}^{\prime}, b_{1:\T}^\prime) \\
    \label{eq:pgas:pfthm1_kernel2}
    &\hspace{2em}\times \delta_{x_1^\prime}(dx_1^{b_{1}^\prime}) \cdots \delta_{x_\T^\prime}(dx_\T^{b_{\T}^\prime}), \hspace{5em}\forall x_{1:\T}^\prime \in \setX^\T, \forall B\in \sigmaX^\T.
  \end{align}
  Plugging \eqref{eq:pgas:pfthm1_kernel2} into \eqref{eq:pgas:pfthm1_invariance1} gives the desired result,
  \begin{align}
    \nonumber
    \ntDist{\theta}{\T}(B) &= \int \kPGAS(x_{1:\T}^\prime, B) \ntDist{\theta}{\T}(x_{1:\T}^\prime)
    \underbrace{\left( \sum_{b_{1:\T}^\prime} \frac{1}{\Np^\T} \right)}_{=1} dx_{1:\T}^\prime, &\forall B&\in\sigmaX^\T.
  \end{align}
\end{proof}

%
%
% =================== ERGODICITY
%
%
\subsection{Ergodicity}

To show ergodicity of the \pgas kernel we need to characterize the support of the target and the proposal densities. Let,
\begin{subequations}
  \begin{align}
    \mathcal{S}_{\theta,t} &= \{x_{1:t} \in \setX^t : \ntDist{\theta}{t}(x_{1:t}) > 0 \}, \\
    \mathcal{Q}_{\theta,t} &= \{x_{1:t} \in \setX^t : \r{\theta}{t}(x_t \mid x_{1:t-1}) \ntDist{\theta}{t-1}(x_{1:t-1}) > 0 \},
  \end{align}
\end{subequations}
with obvious modifications for $t=1$. The following is a minimal assumption.
\begin{asmp}
  \label{asmp:pgas:proposals}
  For any $\theta \in \Theta$ and $t \in \crange{1}{\T}$ we have $\mathcal{S}_t^\theta \subseteq \mathcal{Q}_t^\theta$.
\end{asmp}

Assumption~(A\ref{asmp:pgas:proposals}) basically states that the support of the proposal density should cover the support of the target density.
Ergodicity of \pg under Assumption~(A\ref{asmp:pgas:proposals}) has been established by Andrieu et al. \cite{AndrieuDH:2010}. The same argument can be
applied also to \pgas.

\begin{theorem}[Andrieu et al. {\cite[Theorem~5]{AndrieuDH:2010}}]
  Assume (A\ref{asmp:pgas:proposals}). Then, for any $\Np \geq 2$ and $\theta \in \Theta$, $\kPGAS$ is $\ntDist{\theta}{\T}$-irreducible and aperiodic. Consequently,
  \begin{align*}
    \lim_{n \goesto \infty} \| (\kPGAS)^n(x_{1:\T}^\prime, \cdot) - \ntDist{\theta}{\T}(\cdot) \|_{\TV} &= 0,
    &&\forall x_{1:\T}^\prime \in \setX^\T.
  \end{align*}
\end{theorem}

To strengthen the ergodicity results for the \pgas kernel, we use a
boundedness condition for the importance weights, given in assumption~(A\ref{asmp:pgas:bounds}) below. Such a condition is typical
also in classical importance sampling and is, basically, a slightly stronger version of assumption~(A\ref{asmp:pgas:proposals}).

\begin{asmp} \label{asmp:pgas:bounds}
 For any $\theta\in\Theta$ and $t \in \crange{1}{\T}$, there exists a constant $\ergokappa < \infty$ such that
 $\supnorm{ \wf{\theta}{t} } \leq \ergokappa$.
\end{asmp}

\begin{theorem} \label{thm:pgas:uniform}
  Assume (A\ref{asmp:pgas:bounds}). Then, for any  $\Np\geq2$ and $\theta\in\Theta$, $\kPGAS$ is uniformly ergodic. That is,
  there exist constants $\ergoR < \infty$ and $\ergorho \in [0,1)$ such that
  \begin{align*}
    \| (\kPGAS)^n(x_{1:\T}^\prime, \cdot) - \ntDist{\theta}{\T}(\cdot) \|_{\TV} \leq \ergoR \ergorho^{n},
    &&\forall x_{1:\T}^\prime \in \setX^\T.
  \end{align*}
\end{theorem}
\begin{proof}
  We show that $\kPGAS$ satisfies a Doeblin condition,
  \begin{align}
    \label{eq:pgas:uni:doeblin}
    \kPGAS(x_{1:\T}^\prime, B) &\geq \ergoepsilon \ntDist{\theta}{\T}(B),
    & &\forall x_{1:\T}^\prime \in \setX^\T, \forall B \in \sigmaX^\T,
  \end{align}
  for some constant $\ergoepsilon > 0$.
  Uniform ergodicity 
  then follows from \cite[Proposition~2]{Tierney:1994}.
  To prove \eqref{eq:pgas:uni:doeblin} we use the representation of the \pgas kernel in \eqref{eq:pgas:kernel_as_expectation},
  \begin{align}
    \nonumber
    &\kPGAS(x_{1:\T}^\prime, B) = \E_{\theta,x_{1:\T}^\prime} \left[ \I_B(x_{1:\T}^k ) \right]
    = \sum_{j=1}^\Np \E_{\theta,x_{1:\T}^\prime} \left[ \frac{w_\T^j}{ \sum_l w_\T^l }\I_B(x_{1:\T}^j ) \right] \\
    \label{eq:pgas:uni:first}
    &\geq \frac{1}{\Np\ergokappa}  \sum_{j=1}^{\Np-1} \E_{\theta,x_{1:\T}^\prime} \left[ w_\T^j \I_B(x_{1:\T}^j ) \right]
    = \frac{\Np-1}{\Np\ergokappa}\E_{\theta,x_{1:\T}^\prime} \left[ \wf{\theta}{\T} (x_{1:\T}^1) \I_B(x_{1:\T}^1 ) \right].
  \end{align}
  Here, the inequality follows from bounding the weights in the normalization by $\ergokappa$ and by simply
  discarding the $\Np$th term of the sum  (which is clearly nonnegative). The last equality
  follows from the fact that the particle trajectories $\{x_{1:\T}^i \}_{i=1}^{\Np-1}$ are equally distributed
  under Algorithm~\ref{alg:pgas:pgas}.

  Let $h_{\theta,t} : \setX^t \mapsto \nonnegatives$ and consider
  \begin{align}
    \nonumber
    \E_{\theta,x_{1:\T}^\prime}&{} \left[ h_{\theta,t}( x_{1:t}^1 ) \right]
    = \E_{\theta,x_{1:\T}^\prime} \left[ \E_{\theta,x_{1:\T}^\prime} \left[ h_{\theta,t}( x_{1:t}^1 )\mid \xx_{1:t-1}, \aa_{2:t-1} \right] \right] \\
    \nonumber
    &=\E_{\theta,x_{1:\T}^\prime}
    \left[ \sum_{j=1}^\Np \int h_{\theta,t}( (x_{1:t-1}^j, x_{t}) ) \frac{w_{t-1}^j}{\sum_l w_{t-1}^l} \r{\theta}{t}(x_t \mid x_{1:t-1}^j)\,dx_t \right] \\
    \label{eq:pgas:uni:second}
    &\geq \frac{\Np-1}{\Np\ergokappa}
    \E_{\theta,x_{1:\T}^\prime} \left[  \int h_{\theta,t}( (x_{1:t-1}^1, x_{t}) ) \wf{\theta}{t-1}( x_{1:t-1}^1) \r{\theta}{t}(x_t \mid x_{1:t-1}^1)\,dx_t \right],
  \end{align}
  where the inequality follows analogously to \eqref{eq:pgas:uni:first}.
  Now, let
  \begin{align*}
    h_{\theta,\T}(x_{1:\T}) &= \wf{\theta}{\T}(x_{1:\T}) \I_B(x_{1:\T}),&&\\
    h_{\theta,t-1}(x_{1:t-1}) &= \int h_{\theta,t}(x_{1:t}) \wf{\theta}{t-1}(x_{1:t-1}) \r{\theta}{t}(x_t \mid x_{1:t-1})\,dx_t, & t&\leq\T.
  \end{align*}
  Then, by iteratively making use of \eqref{eq:pgas:uni:second} and changing the order of integration, we can bound \eqref{eq:pgas:uni:first}
  according to
  \begin{align}
    \nonumber
    &\left(\frac{\Np-1}{\Np\ergokappa}\right)^{-\T} \kPGAS(x_{1:\T}^\prime, B) \geq \E_{\theta,x_{1:\T}^\prime} \left[ h_{\theta,1}(x_1^1) \right]\\
    \nonumber
    &\hspace{2em}= \int \wf{\theta}{1}(x_1) \r{\theta}{1}(x_1)
    \prod_{t=2}^\T \left( \wf{\theta}{t}(x_{1:t}) \r{\theta}{t}(x_t \mid x_{1:t-1})\right) \I_B(x_{1:\T}) \,dx_{1:\T} \\
    \nonumber
    &\hspace{2em}=   \int \tDist{\theta}{1}(x_1) \prod_{t=2}^\T \left(
    \frac{ \tDist{\theta}{t}(x_{1:t})}{\tDist{\theta}{t-1}(x_{1:t-1}) } \right)\I_B(x_{1:\T}) \,dx_{1:\T} \\
    &\hspace{2em}= \int \tDist{\theta}{\T}(x_{1:\T})\I_B(x_{1:\T}) \,dx_{1:\T}
    = \normconst{\theta}{\T} \ntDist{\theta}{\T}(B).
  \end{align}
  With $\Np\geq2$ and since $\normconst{\theta}{\T}>0$ the result follows.
\end{proof}

%%%%%%%%%%%%%%%%%%%%%%%
%     NEW SECTION     %
%%%%%%%%%%%%%%%%%%%%%%%
\section{\pgas for state-space models}\label{sec:ssm}%

\subsection{Learning of state-space models with \pgas}
\ssm{s} comprise an important special case of the model class treated above. In this section, we
illustrate how \pgas can be used for inference and learning of these models.
We consider the nonlinear/non-Gaussian \ssm
\begin{subequations}
  \label{eq:ssm:model}
  \begin{align}
    x_{t+1} &\sim f_\theta(x_{t+1} \mid x_t), \\
    y_t &\sim g_\theta(y_t \mid x_t),
  \end{align}
\end{subequations}
and $x_1 \sim \mu_\theta(x_1)$, where $\theta \in \Theta$ is a static parameter, $x_t$ is the
latent state and $y_t$ is the observation at time $t$, respectively.
Given a batch of measurements $y_{1:\T}$, we wish to make inferences about $\theta$ and/or about the latent states $x_{1:\T}$.

Consider first the Bayesian setting where a prior distribution $\pi(\theta)$ is assigned to $\theta$.
We seek the parameter posterior $p(\theta \mid x_{1:\T})$ or, more generally, the joint state and parameter posterior
${p(\theta, x_{1:\T} \mid y_{1:\T})}$. Gibbs sampling can be used to simulate from this distribution by sampling
the state variables $\{x_t\}$ one at a time and the parameters $\theta$ from their respective posteriors.
However, it has been recognized that this can result in poor mixing, due to the often high autocorrelation of the
state sequence. The \pgas kernel offers a different approach, namely to sample the complete state
trajectory $x_{1:\T}$ in one block. This can considerably improve the mixing of the sampler \cite{JongS:1995}.
Due to the invariance property of the kernel (Theorem~\ref{thm:pgas:invariance}),
the validity of the Gibbs sampler is not violated. We summarize the procedure in Algorithm~\ref{alg:ssm:pgas}.

\begin{algorithm}[ptb]
  \caption{\pgas for Bayesian learning of \ssm{s}}
  \label{alg:ssm:pgas}
  \begin{algorithmic}[1]
    \STATE Set $\theta[0]$ and $x_{1:\T}[0]$ arbitrarily.
    \FOR{$n \geq 1$}
    \STATE Draw $x_{1:\T}[n] \sim \kPGAS[\theta{[n-1]}](x_{1:\T}[n-1], \cdot)$. \textit{/* By running Algorithm~\ref{alg:pgas:pgas} */}
    \STATE Draw $\theta[n] \sim p(\theta \mid x_{1:\T}[n], y_{1:\T})$.
    \ENDFOR
  \end{algorithmic}
\end{algorithm}

\pgas is also useful for maximum-likelihood-based learning of \ssm{s}.
A popular strategy for computing the maximum likelihood estimator
\begin{align}
  \widehat\theta_{\text{ML}} = \argmax_{\theta\in\Theta} \log p_\theta(y_{1:\T})
\end{align}
is to use the expectation maximization (\exm) algorithm \cite{DempsterLR:1977,McLachlanK:2008}.
\exm is an iterative method, which maximizes $\log p_\theta(y_{1:\T})$
by iteratively maximizing an auxiliary quantity:
%\begin{align}
$  \theta[n] = \argmax_{\theta\in\Theta} Q(\theta, \theta[n-1])$,
%\end{align}
where
\begin{align}
  Q(\theta, \theta[n-1]) = \int \log p_\theta(x_{1:\T},y_{1:\T}) p_{\theta[n-1]}(x_{1:\T} \mid y_{1:\T}) \,dx_{1:\T}.
\end{align}
When the above integral is intractable to compute, one can use a Monte Carlo approximation or a stochastic approximation
of the intermediate quantity, leading to the \mcem \cite{WeiT:1990} and the \saem \cite{DelyonLM:1999} algorithms, respectively.
When the underlying Monte Carlo simulation is computationally involved, SAEM is particularly useful since it makes efficient
use of the simulated values.
The SAEM approximation of the auxiliary quantity is given by
\begin{align}
  \label{eq:ssm:Qsaem}
  \widehat Q_n(\theta) = (1-\steplength_n)\widehat Q_{n-1}(\theta) + \steplength_n \left( \log p_\theta(x_{1:\T}[n], y_{1:\T}) \right),
\end{align}
where $\steplength_n$ is the step size and, in the vanilla form of SAEM,
 $x_{1:\T}[n]$ is drawn from the joint smoothing density $p_{\theta[n-1]}(x_{1:\T} \mid y_{1:\T})$.
In practice, the stochastic approximation update \eqref{eq:ssm:Qsaem} is typically made
on some sufficient statistic for the complete data log-likelihood; see \cite{DelyonLM:1999} for details.
While the joint smoothing density is intractable for a general nonlinear/non-Gaussian \ssm,
it has been recognized that it
is sufficient to sample from a uniformly ergodic Markov kernel, leaving the joint smoothing distribution invariant \cite{BenvenisteMP:1990,AndrieuMP:2005}.
A practical approach is therefore to compute the auxiliary quantity according to the stochastic approximation \eqref{eq:ssm:Qsaem},
but where $x_{1:\T}[n]$ is simulated from the \pgas kernel $\kPGAS[\theta{[n-1]}](x_{1:\T}[n-1], \cdot)$.
This particle SAEM algorithm, previously presented in \cite{Lindsten:2013}, is summarized in Algorithm~\ref{alg:ssm:psaem}.

\begin{algorithm}[ptb]
  \caption{\pgas for frequentist learning of \ssm{s}}
  \label{alg:ssm:psaem}
  \begin{algorithmic}[1]
    \STATE Set $\theta[0]$ and $x_{1:\T}[0]$ arbitrarily. Set $\widehat Q_0(\theta) \equiv 0$.
    \FOR{$n \geq 1$}
    \STATE Draw $x_{1:\T}[n] \sim \kPGAS[\theta{[n-1]}](x_{1:\T}[n-1], \cdot)$.  \textit{/* By running Algorithm~\ref{alg:pgas:pgas} */}
    \STATE Compute $\widehat Q_n(\theta)$ according to \eqref{eq:ssm:Qsaem}. % \textit{/* By updating the sufficient statistics */}
    \STATE Compute $\theta[n] = \argmax_{\theta\in\Theta} \widehat Q_n(\theta)$.
    \IF{convergence criterion is met}
    \RETURN $\theta[n]$.
    \ENDIF
    \ENDFOR
  \end{algorithmic}
\end{algorithm}

\subsection{Sampling from the \pgas kernel}

Sampling from the \pgas kernel, \ie, running Algorithm~\ref{alg:pgas:pgas}, is similar to running a PF.
The only non-standard (and nontrivial) operation is the ancestor sampling step. For the learning
algorithms discussed above, the distribution of interest is the joint smoothing distribution.
Consequently, the unnormalized target density is given by
$\tDist{\theta}{\T}(x_{1:\T}) = p_\theta(x_{1:\T}, y_{1:\T})$. The ancestor sampling weights in \eqref{eq:pgas:as_weights}
are thus given by
\begin{align}
  \label{eq:ssm:weights}
  \widetilde w_{t-1\mid\T}^i = w_{t-1}^i p_\theta(x_{t:\T}^\prime, y_{t:\T} \mid x_{1:t-1}^i, y_{1:t-1}) \propto
  w_{t-1}^i f_\theta(x_{t}^\prime \mid x_{t-1}^i).
\end{align}
This expression can be understood as an application of Bayes' theorem. The importance weight $w_{t-1}^i$ is the prior probability
of the particle $x_{t-1}^i$ and the factor $f_\theta(x_{t}^\prime \mid x_{t-1}^i)$ is the likelihood of
moving from $x_{t-1}^i$ to $x_t^\prime$. The product of these two factors is thus proportional to the
posterior probability that $x_t^\prime$ originated from $x_{t-1}^{i}$.

The expression \eqref{eq:ssm:weights} can also be recognized as a
one-step backward simulation; see \cite{GodsillDW:2004,LindstenS:2013}.
This highlights the close relationship between \pgas and \pgbs. The latter method
is conceptually similar to \pgas, but it make use of an explicit backward simulation pass; see \cite{Whiteley:2010,WhiteleyAD:2010} or \cite[Section~5.4]{LindstenS:2013}.
More precisely, to generate a draw from the \pgbs kernel, we first run a particle filter with reference trajectory $x_{1:\T}^\prime$
\emph{without} ancestor sampling (\ie, in Algorithm~\ref{alg:pgas:pgas}, we replace line~\ref{row:atN}
with $a_t^\Np = \Np$, as in the basic \pg sampler). Thereafter, we extract a new trajectory by running a backward simulator.
That is, we draw $j_{1:\T}$ with $\Prb(j_\T = i) \propto w_\T^i$ and then, for $t = \T-1$ to $1$,
\begin{align}
  \label{eq:equiv:bsi}
  \Prb(j_t = i \mid j_{t+1}) \propto w_t^i f_{\theta}(x_{t+1}^{j_{t+1}} \mid x_t^{i}),
\end{align}
and take $x_{1:\T}^\star = x_{1:\T}^{j_{1:\T}}$ as the output from the algorithm. In the above, the conditioning
on the forward particle system $\{ \xx_{1:\T}, \aa_{2:\T} \}$ is implicit.

Let the Markov kernel on $(\setX^\T, \sigmaX^\T)$ defined by this procedure be denoted as $\kPGBS$.
An interesting question to ask is whether or not the \pgas kernel $\kPGAS$ and the \pgbs kernel $\kPGBS$ are probabilistically equivalent.
It turns out that, in some specific scenarios, this is indeed the case.

\begin{proposition}\label{prop:equiv}%
  Assume that \pgas and \pgbs both target the joint smoothing distribution for an \ssm and that both methods use
  the bootstrap proposal kernel in the internal particle filters,
  \ie, $\r{\theta}{t}(x_t \mid x_{1:t-1}) = f_\theta(x_t \mid x_{t-1})$. Then, for any $x_{1:\T}^\prime\in\setX^\T$ and $B \in \sigmaX^\T$,
  $\kPGAS(x_{1:\T}^\prime, B) = \kPGBS(x_{1:\T}^\prime, B)$.
\end{proposition}
\begin{proof}
  See Appendix~\ref{app:proof}.
\end{proof}

Proposition~\ref{prop:equiv} builds upon \cite[Proposition~5]{OlssonR:2011}, where the equivalence between a
(standard) bootstrap \pf and a backward simulator is established. In Appendix~\ref{app:proof}, we adapt their argument to handle
the case with conditioning on a reference trajectory and ancestor sampling.
The conditions of Proposition~\ref{prop:equiv} imply that the weight functions
\eqref{eq:smc:weightfunction} in the internal particle filters are independent of the ancestor indices.
This is key in establishing the above result and we emphasize that the equivalence between the samplers
does not hold in general for models outside the class of \ssm{s}. In particular,
for the class of non-Markovian latent variable models, discussed in the subsequent section,
we have found that the samplers have quite different properties.

%%%%%%%%%%%%%%%%%%%%%%%
%     NEW SECTION     %
%%%%%%%%%%%%%%%%%%%%%%%
\section{\pgas for non-Markovian models}\label{sec:nonmarkov}%

\subsection{Non-Markovian latent variable models}\label{sec:nonmarkov:applications}
A very useful generalization of \ssm{s} is the class of non-Markovian latent variable models,
\begin{subequations}
  \label{eq:nonmarkov:model}
  \begin{align}
    x_{t+1} &\sim f_\theta(x_{t+1} \mid x_{1:t}), \\
    y_t &\sim g_\theta(y_t \mid x_{1:t}).
  \end{align}
\end{subequations}
%with initial density $\initialdensity(x_1)$.
Similarly to the \ssm \eqref{eq:ssm:model}, this model is characterized by
a latent process $x_t \in \setX$ and an observed process $y_t \in \setY$.
However, it does not share the conditional independence properties that are central to \ssm{s}. Instead,
both the transition density $f_\theta$ and the measurement density $g_\theta$ may depend on the entire past history of the latent process.
In Sections~\ref{sec:nonmarkov:mh}~and~\ref{sec:nonmarkov:truncation}, we discuss the
ancestor sampling step of the \pgas algorithm specifically for these non-Markovian models. We
consider two approaches for efficient implementation of this step, first by using Metropolis-Hastings
within \pgas and then by using a truncation strategy for the ancestor sampling weights.
First, however, to motivate the present development we review some application areas in which this type of models arise.

In Bayesian nonparametrics~\cite{HjortHMW:2010} the
latent random variables of the classical Bayesian model are replaced by latent stochastic processes, which
are typically non-Markovian.
This includes popular models based on the Dirichlet process, \eg, \cite{TehJBB:2006,EscobarW:1995}, and Gaussian process regression
and classification models \cite{RasmussenW:2006}. These processes are also commonly used as components in
hierarchical Bayesian models, which then inherit their non-Markovianity. An example is the Gaussian
process \ssm \cite{TurnerD:2010,FrigolaLSR:2013}, a flexible nonlinear dynamical systems model, for which
\pgas has been successfully applied \cite{FrigolaLSR:2013}.

% Marginalization of SSMs
Another typical source of non-Markovianity is by marginalization over part of the state vector
(\ie, Rao-Blackwellization \cite{ChenL:2000,WhiteleyAD:2010,LindstenBGS:2013}) or by a change of variables in an \ssm.
This type of operations typically results in a loss of the Markov property, but
they can, however, be very useful. For instance, by expressing an \ssm in terms of its ``innovations'' (\ie, the driving noise
of the state process), it is possible to use backward and ancestor sampling in models for which the
state transition density is not available to us. This includes many models for which the transition is implicitly given by a simulator
\cite{GanderS:2007,FearnheadPR:2008,GolightlyW:2008,MurrayJP:2012}
or degenerate models where the transition density does not even exist \cite{RisticAG:2004,GustafssonGBFJKN:2002}.
We illustrate these ideas in Section~\ref{sec:eval}. See also \cite[Section~4]{LindstenS:2013} for
a more in-depth discussion on reformulations of \ssm{s} as non-Markovian models.

% Graphical models
Finally, it is worth to point out that many statistical models which are not sequential ``by nature'' can be
conveniently viewed as non-Markovian latent variable models. This includes, among others, probabilistic graphical
models such as Markov random fields; see \cite[Section~4]{LindstenS:2013}.

\subsection{Forced move Metropolis-Hastings}\label{sec:nonmarkov:mh}
To employ \pgas (or in fact any backward-simulation-based method; see \cite{LindstenS:2013}) we need to evaluate the ancestor sampling weights \eqref{eq:pgas:as_weights} which depend on the ratio,
\begin{align}
  \label{eq:nonmarkov:ratio}
  \frac{ \tDist{\theta}{\T}(x_{1:\T}) }{ \tDist{\theta}{t-1}(x_{1:t-1}) } = \frac{p_\theta(x_{1:\T}, y_{1:\T})}{p_\theta(x_{1:t-1}, y_{1:t-1})} %= p(x_{t+1:\T}, y_{t+1:\T} \mid x_{1:t}, y_{1:t})
  = \prod_{s = t}^\T g_\theta(y_{s} \mid x_{1:s}) f_\theta(x_s \mid x_{1:s-1}).
\end{align}
Assuming that $g_\theta$ and $f_\theta$ can both be evaluated in constant time,
the computational cost of computing the backward sampling weights \eqref{eq:pgas:as_weights} will thus
be $\Ordo(\Np\T)$. This step can easily become the computational bottleneck when applying the PGAS algorithm to a non-Markovian model.

A simple way to reduce the complexity is to
employ Metropolis-Hastings (MH) within \pgas. Let
\begin{align}
  \label{eq:nonmarkov:asprob}
  \asprob(k) = \frac{ \widetilde w_{t-1\mid\T}^k }{ \sum_{l=1}^\Np \widetilde w_{t-1\mid\T}^l}
\end{align}
denote the law of the ancestor index $a_t^\Np$, sampled at line~\ref{row:atN} of Algorithm~\ref{alg:pgas:pgas}.
From Lemma~\ref{lem:pgas:invariance}, we know that this step of the algorithm in fact corresponds to a Gibbs step for the
extended target distribution \eqref{eq:pgas:phidef}. To retain the correct limiting distribution
of the \pgas kernel, it is therefore sufficient that $a_t^\Np$ is sampled from a Markov kernel leaving \eqref{eq:nonmarkov:asprob}
invariant (resulting in a standard combination of \mcmc kernels; see, \eg, \cite{Tierney:1994}).

Let $q(k^\prime \mid k)$ be an MH proposal kernel on $\crange{1}{\Np}$. We can thus propose a move for the ancestor index $a_t^\Np$,
from $\Np$ to $k^\prime$, by simulating $k^\prime \sim q(\cdot \mid \Np)$. With probability
\begin{align}
  \label{eq:nonmarkov:acceptprob}
  1\wedge \frac{\widetilde w_{t-1\mid\T}^{k^\prime}  }{ \widetilde w_{t-1\mid\T}^\Np  }\frac{ q(\Np \mid k^\prime) }{ q(k^\prime \mid \Np) }
\end{align}
the sample is accepted and we set $a_t^\Np = k^\prime$, otherwise we keep the ancestry $a_t^\Np = \Np$.
Using this approach, we only need to evaluate the ancestor sampling weights for the proposed values, bringing the total
computational cost down from $\Ordo(\Np\T^2)$ to $\Ordo(\Np\T + \T^2)$. While still quadratic in $\T$, this reduction can be very useful
whenever $\Np$ is moderately large.

Since the variable $a_t^\Np$ is discrete-valued, it is recommended to use a \emph{forced move} proposal
in the spirit of \cite{Liu:1996a}. That is, $q$ is constructed so that $q(k\mid k) = 0, \forall k$, ensuring
that the current state of the chain is not proposed anew, which would be a wasteful operation.
One simple choice is to let $q(k^\prime \mid k)$ be uniform over $\crange{1}{\Np}\setminus k$.
In the subsequent section, we discuss a different strategy for reducing the complexity of the
ancestor sampling step, which can also be used to design a better proposal for the forced move MH sampler.

\subsection{Truncation of the ancestor sampling weights}\label{sec:nonmarkov:truncation}
The quadratic computational complexity in $\T$ for the forced move MH sampler may still be
prohibitive if $\T$ is large. To make progress, we consider non-Markovian models in which there is a
decay in the influence of the past on the present, akin to that in 
Markovian models but without the strong Markovian assumption. Hence, it is possible
to obtain a useful approximation of the ancestor sampling weights by truncating the product \eqref{eq:nonmarkov:ratio}
to a smaller number of factors, say $\ell$. We can thus replace \eqref{eq:pgas:as_weights} with the approximation
\begin{align}
  \nonumber
  \widetilde w_{t-1\mid\T}^{\ell,i} &\triangleq w_{t-1}^i \frac{\tDist{\theta}{t-1+\ell}( (x_{1:t-1}^i, x_{t:t-1+\ell}^\prime ) )}{ \tDist{\theta}{t-1} (x_{1:t-1}^i) } \\  \label{eq:nonmarkov:truncated_weights}
&=  w_{t-1}^i \prod_{s = t}^{t-1+\ell} g_\theta(y_{s} \mid x_{1:t-1}^i, x_{t:s}^\prime) f_\theta(x_s^\prime \mid x_{1:t-1}^i, x_{t:s-1}^\prime).
\end{align}
Let $\asprobtrunc[\ell](k)$ be the probability distribution defined by the truncated ancestor sampling weights \eqref{eq:nonmarkov:truncated_weights},
analogously to \eqref{eq:nonmarkov:asprob}. The following proposition formalizes our assumption.
\begin{proposition}\label{prop:truncation_kld}%
  Let
  \begin{align*}
    h_s(k) = g_\theta(y_{t-1+s} \mid x_{1:t-1}^k, x_{t:t-1+s}^\prime) f_\theta(x_{t-1+s}^\prime \mid x_{1:t-1}^k, x_{t:t-1+s}^\prime)
  \end{align*}
  and assume that
  $\max_{k,l} \left( h_s(k) / h_s(l) -1 \right) \leq A\exp(-cs)$,
  for some constants $A$ and $c > 0$. Then, $\KLD(\asprob \| \asprobtrunc) \leq C\exp(-c\ell)$ for some constant $C$, where $\KLD$ is the Kullback-Leibler (KL) divergence.
\end{proposition}
\begin{proof}
  See Appendix~\ref{app:proof}.
\end{proof}
Using the approximation given by \eqref{eq:nonmarkov:truncated_weights}, the
ancestor sampling weights can be computed in constant time within the \pgas framework.
The resulting approximation can be quite useful; indeed, in our experiments we have
seen that even $\ell = 1$ can lead to very accurate inferential results. In general,
however, it will not be known \emph{a priori} how to set the truncation level $\ell$.
To address this problem, we propose to use an adaption of the truncation 
level. Since the approximative weights \eqref{eq:nonmarkov:truncated_weights} 
can be evaluated sequentially, the idea is to start with $\ell = 1$ and then increase 
$\ell$ until the weights have, in some sense, converged. In particular, in our
experimental work, we have used the following simple approach.

At time $t$, let $\varepsilon_\ell = \DTV(\asprobtrunc[\ell], \asprobtrunc[\ell-1])$
be the total variation (TV) distance between the approximative ancestor sampling distributions for two consecutive truncation levels.
We then compute the exponentially decaying moving average
of the sequence $\varepsilon_\ell$, with forgetting factor $\upsilon \in [0,\,1]$, and stop when this falls below some threshold $\tau \in [0,\,1]$.
This adaption scheme removes the requirement to specify $\ell$ directly, but instead introduces the design parameters $\upsilon$ and $\tau$.
However, these parameters are much easier to reason about---a small value for $\upsilon$ gives a rapid response to changes in $\varepsilon_\ell$
whereas a large value gives
a more conservative stopping rule, improving the accuracy of the approximation at the cost of higher computational complexity.
A similar tradeoff holds for the threshold $\tau$ as well. Most importantly, we have found that the same values
for $\upsilon$ and $\tau$ can be used for a wide range of models, with very different mixing properties.

To illustrate the effect of the adaption rule, and how the distribution $\asprobtrunc$ typically evolves as we increase $\ell$, we provide
two examples in Figure~\ref{fig:degen_Papprox}. These examples are taken from the simulation study provided in Section~\ref{sec:eval:degen_rndlgss}.
Note that the untruncated distribution $\asprob$ is given for the maximal value of $\ell$, \ie, furthest to the right in the figures.
By using the adaptive truncation, we can stop the evaluation of the weights at a much earlier stage, and still obtain an accurate
approximation of $\asprob$.

\begin{figure}[ptb]
  \centering
  \includegraphics[height = \Papproxplotheight]{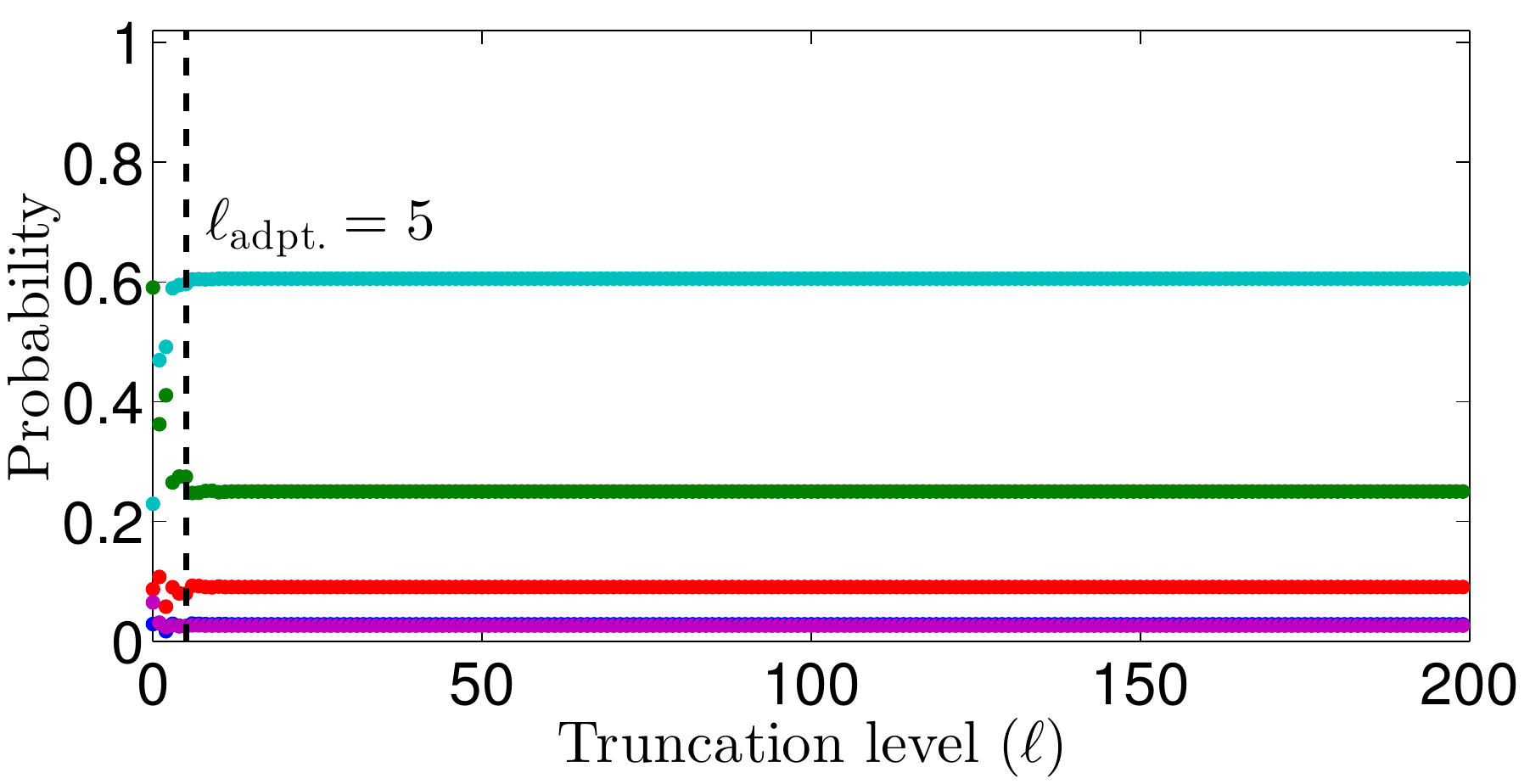}
  \includegraphics[height = \Papproxplotheight]{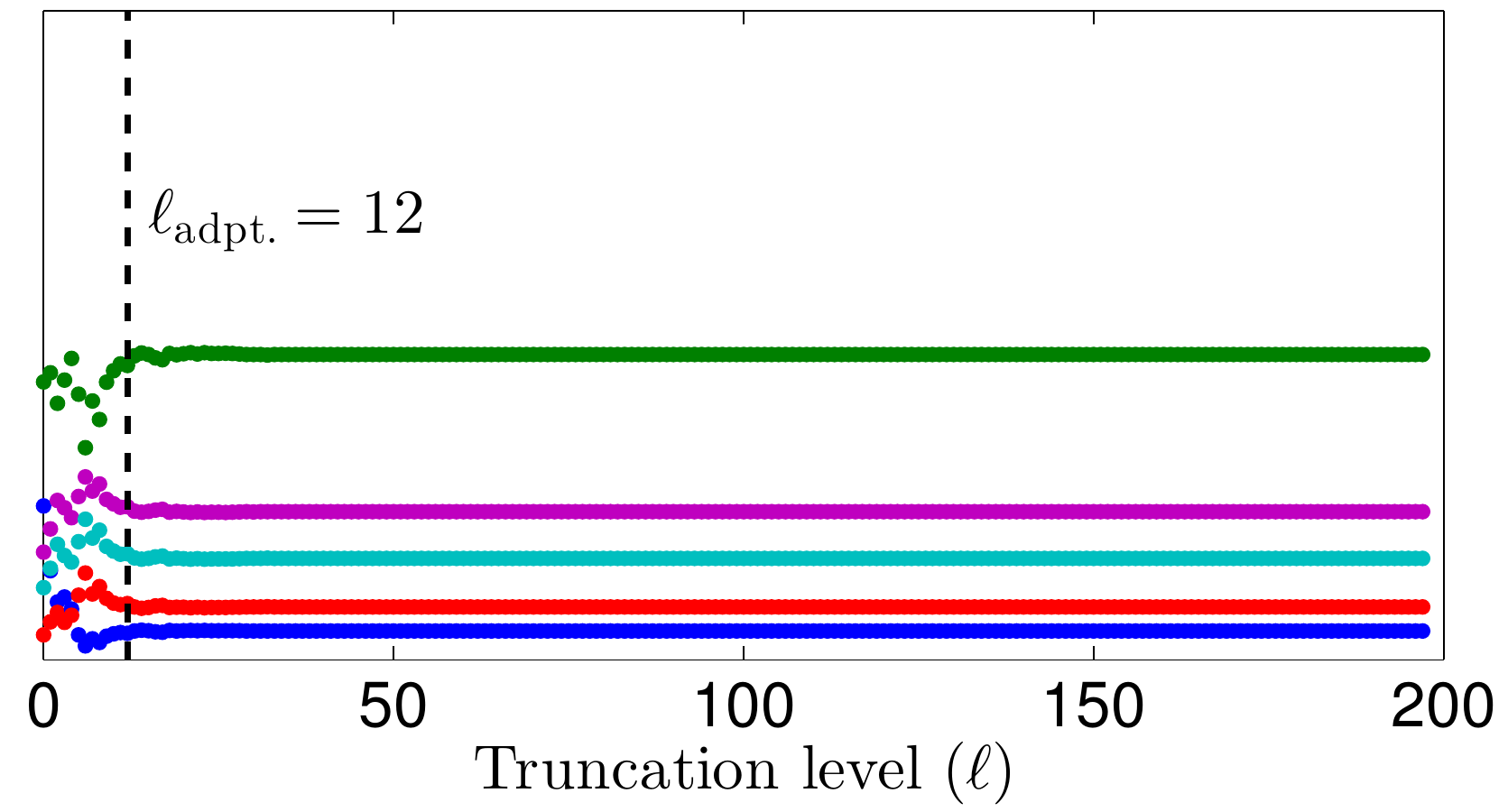}
  \caption{Probability under $\asprobtrunc$ as a function of the truncation level $\ell$
    for two different systems; one 5 dimensional (left) and one 20 dimensional (right).
    The $\Np = 5$ dotted lines correspond to $\asprobtrunc(k)$ for $k \in \crange{1}{\Np}$, respectively (\NB two of the lines overlap in the left figure).
    The dashed vertical lines show the
    value of the truncation level $\ell_{\mathrm{adpt.}}$, resulting from the adaption scheme with $\upsilon = 0.1$ and $\tau = 10^{-2}$.
    See Section~\ref{sec:eval:degen_rndlgss} for details on the experiments.}
  \label{fig:degen_Papprox}
\end{figure}

The approximation \eqref{eq:nonmarkov:truncated_weights} can be used in a few different ways. First, as discussed above,
we can simply replace $\asprob$ with $\asprobtrunc$ in the \pgas algorithm, resulting in a total computational cost of $\Ordo(\Np\T\ell)$.
This is the approach that we have favored, owing to its simplicity and the fact that we have found the truncation to
lead to very accurate approximations. Another approach, however, is to use $\asprobtrunc$ as an efficient proposal distribution
for the MH algorithm suggested in Section~\ref{sec:nonmarkov:mh}. The MH accept/reject decision will then compensate
for the approximation error caused by the truncation. A third approach is to use the MH algorithm, but
to make use of the approximation \eqref{eq:nonmarkov:truncated_weights} when evaluating the acceptance probability \eqref{eq:nonmarkov:acceptprob}.
By doing so, the algorithm can be implemented with $\Ordo(\Np\T + \T\ell)$ computational complexity.

%%%%%%%%%%%%%%%%%%%%%%%
%     NEW SECTION     %
%%%%%%%%%%%%%%%%%%%%%%%
\section{Numerical evaluation}\label{sec:eval}%
In this section we illustrate the properties of \pgas in a simulation study.
First, in Section~\ref{sec:eval:lg1d} we consider a simple linear Gaussian \ssm and investigate
the improvement in mixing offered by ancestor sampling when \pgas is compared with \pg.
We do not consider \pgbs in this example since, by Proposition~\ref{prop:equiv},
\pgas and \pgbs are probabilistically equivalent in this scenario.

When applied to non-Markovian models, however, Proposition~\ref{prop:equiv} does not apply
since the weight function will depend on the complete
history of the particles. \pgas and \pgbs will then have different properties as is
illustrated empirically in Section~\ref{sec:eval:degen_rndlgss}
where we consider inference in degenerate \ssm{s} reformulated as non-Markovian models.
Finally, in Section~\ref{sec:eval:flu} we use a similar reformulation and apply \pgas for identification
of an epidemiological model for which the transition kernel is not available.

\subsection{1st order \lgss model}\label{sec:eval:lg1d}
Consider a first-order linear Gaussian state-space (\lgss) model,
\begin{subequations}
  \label{eq:eval:lg1d_sys}
  \begin{align}
    x_{t+1} &= ax_t + v_t, & v_t&\sim\N(0,q), \\
    y_t &= x_t + e_t, & e_t&\sim\N(0,r),
  \end{align}
\end{subequations}
with initial state $p(x_1) = \N(x_1 ; 0, q/(1-a^2))$ and unknown parameters $\theta = (a,q,r)$.
This system is used as a proof of concept to illustrate the superior mixing of \pgas when compared to \pg.
For this system it is possible to implement an \emph{ideal Gibbs} sampler, \ie, by iteratively
sampling from the posterior parameter distribution $p(\theta \mid x_{1:\T}, y_{1:\T})$ and
from the full joint smoothing distribution $p_\theta(x_{1:\T} \mid y_{1:\T})$. This is useful
for comparison, since the ideal Gibbs sampler is the baseline for both \pg samplers.

We simulate the system \eqref{eq:eval:lg1d_sys} for $T = 100$ time steps with $\theta = (0.8, 1, 0.5)$.
We then run the \pg \cite{AndrieuDH:2010} and the \pgas (Algorithm~\ref{alg:ssm:pgas}) samplers with different number of particles, $\Np \in \{5,20,100,1000\}$,
as well as the ideal Gibbs sampler. All methods are initialized at $\theta[0] = (-0.8, 0.5, 1)$ and
simulated for \thsnd{50} iterations, whereafter the first \thsnd{10} samples are discarded as burn-in.
To evaluate the mixing of the samplers, we compute the autocorrelation functions (\acf{s}) for
the sequences $\theta[n] - \E[\theta \mid y_{1:\T}]$\footnote{The ``true'' posterior mean is computed as the sample
mean obtained from the ideal Gibbs sampler.}. The results for the parameter $q$ are reported in Figure~\ref{fig:lg1d_acf} (top row).
Similar results hold for $a$ and $r$ as well. We see that the \pg sampler requires a large $\Np$ to obtain good mixing.
For $\Np = 100$ the \acf drops off much slower than for the ideal sampler and for $\Np \leq 20$ the \acf
is more or less constant. For \pgas, on the other hand, the \acf is much more robust to the choice of $\Np$.
Indeed, we obtain a mixing which is comparable to that of the ideal Gibbs sampler for any number of particles $\Np \geq 5$.

To further investigate the robustness of \pgas we repeat the same experiment with a larger data batch consisting of $\T = \thsnd{2}$ samples. The results
are given in Figure~\ref{fig:lg1d_acf} (bottom row). The effect can be seen even more clearly in this more challenging scenario.
The big difference in mixing between the two samplers can be understood as a manifestation of how
they are affected by path degeneracy. These results are in agreement with the discussion in Section~\ref{sec:pgas:degeneracy}.

\begin{figure}[ptb]
  \centering
  \includegraphics[width=0.5\columnwidth]{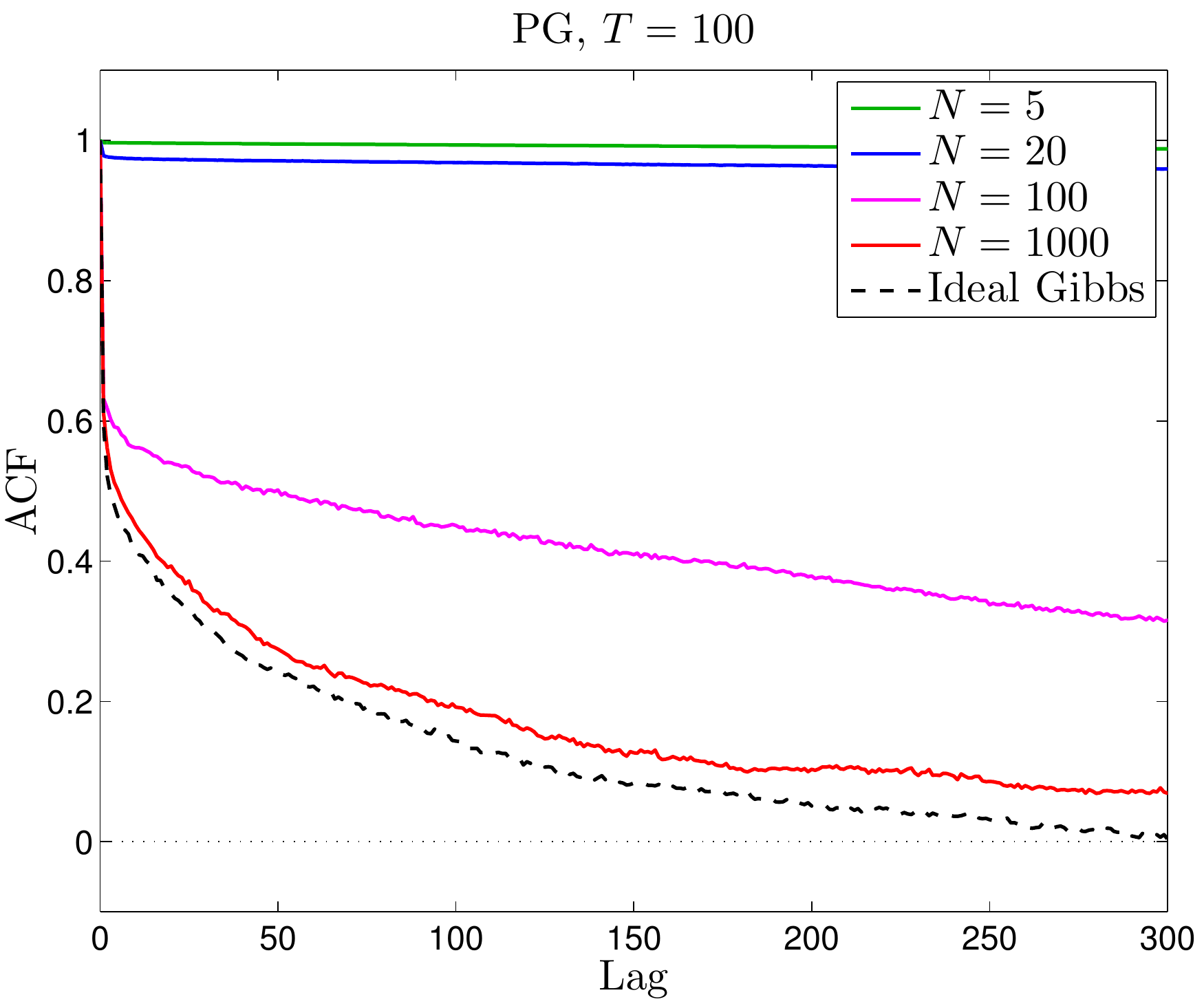}%
  \includegraphics[width=0.5\columnwidth]{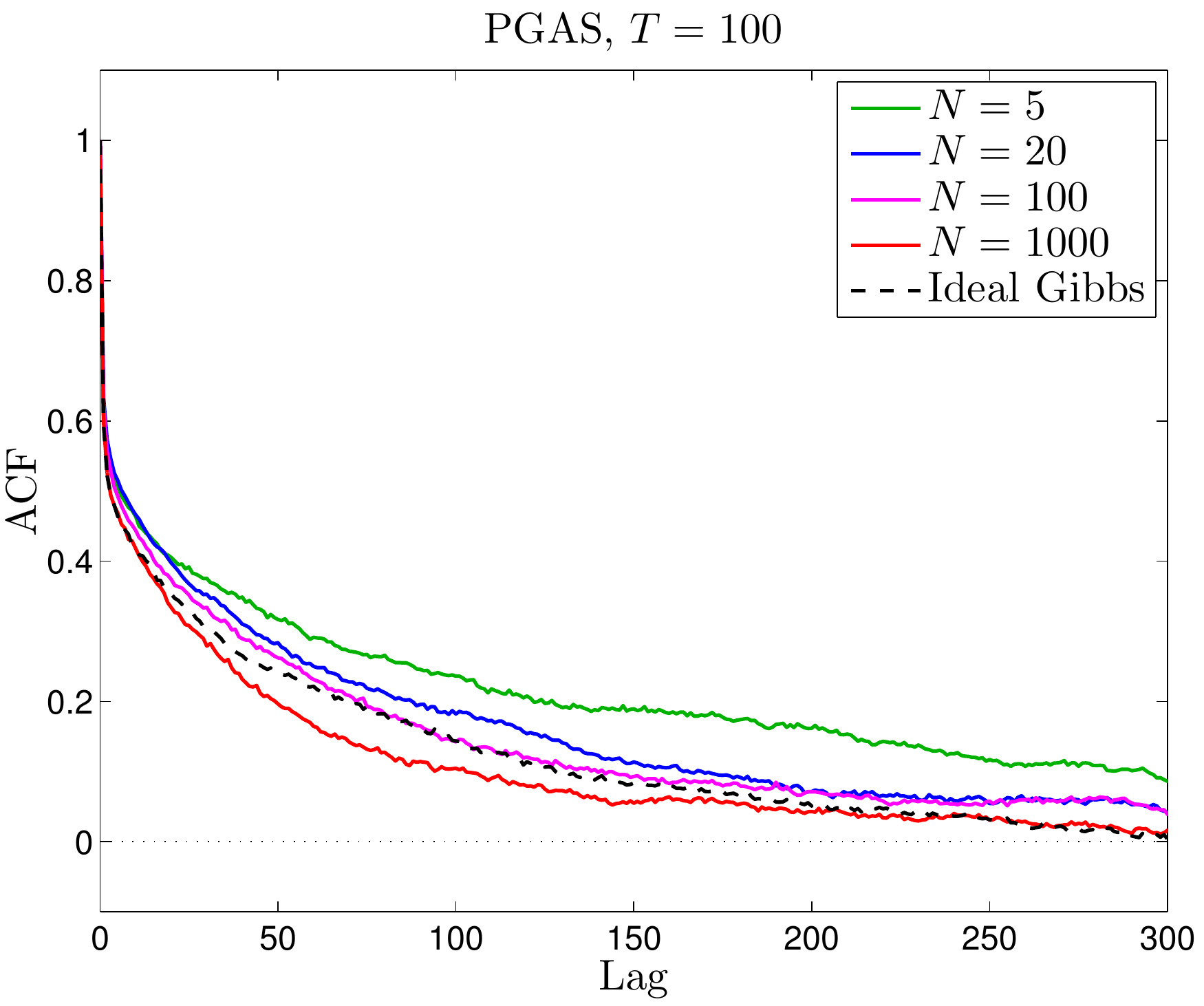}\\
  \includegraphics[width=0.5\columnwidth]{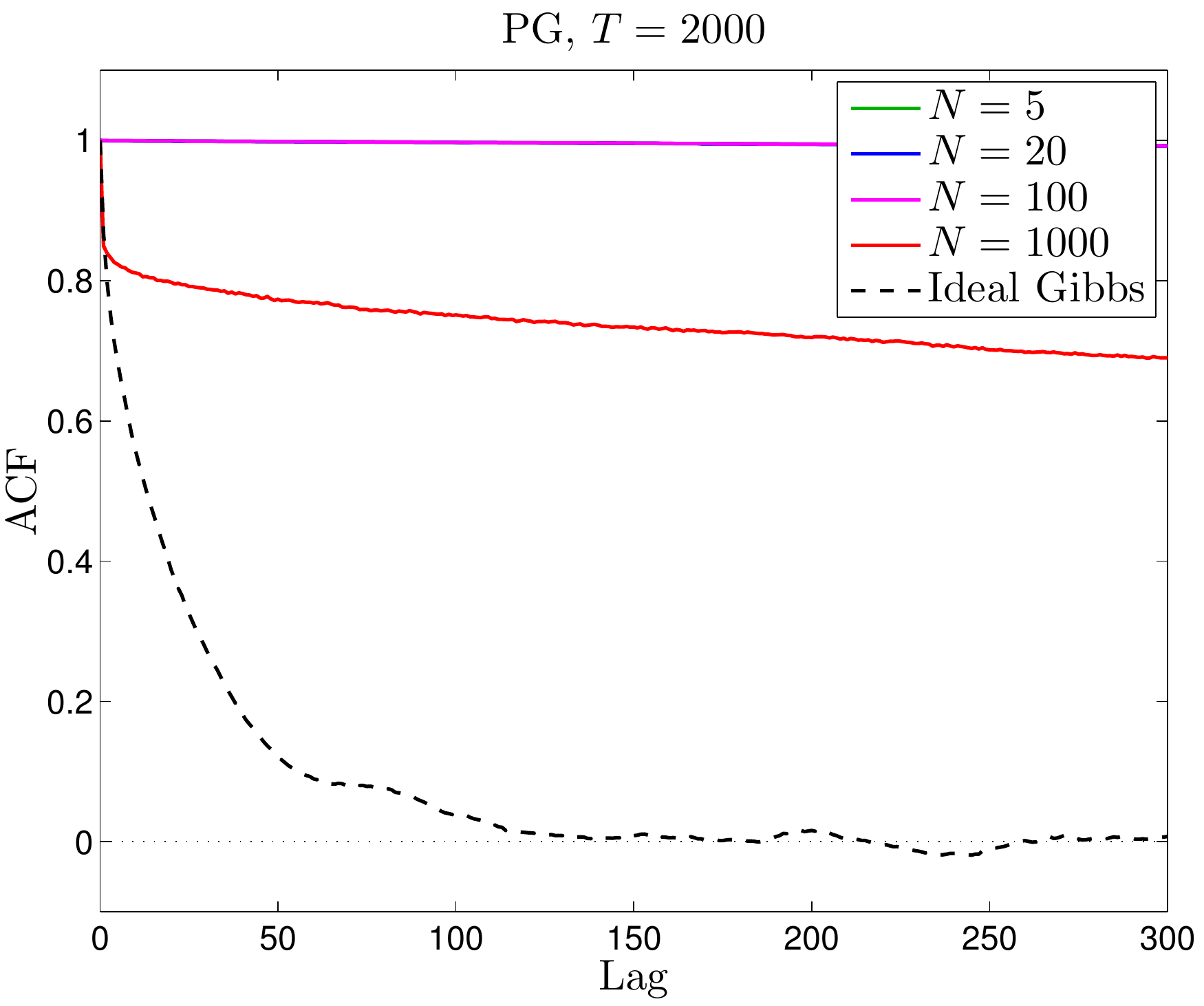}%
  \includegraphics[width=0.5\columnwidth]{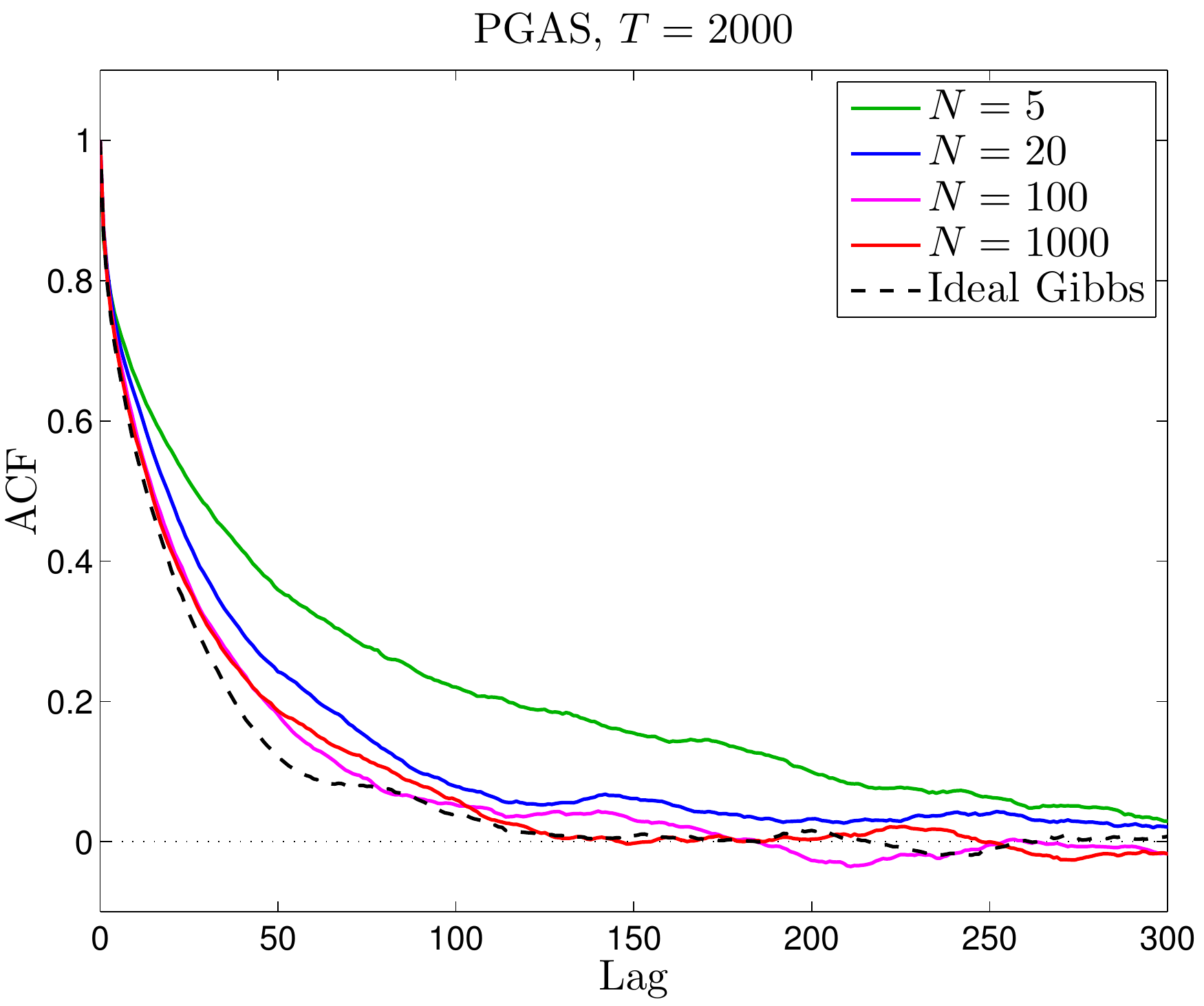}%
  \caption{\acf{s} the parameter $q$ for \pg (left column) and for \pgas (right column) for $T = 100$ (top row) and $T = \thsnd{2}$ (bottom row).
  The results are reported for different number of particles $\Np$, as well as for the ideal Gibbs sampler.      (This figure is best viewed in color.)}
  \label{fig:lg1d_acf}
\end{figure}

\subsection{Degenerate \lgss models}\label{sec:eval:degen_rndlgss}
Many dynamical systems are most naturally modeled as \emph{degenerate} in the sense that
the transition kernel of the state-process does not admit any density
with respect to a dominating measure. It is problematic to use
(particle-filter-based) backward sampling methods for these models, owing to the fact that the backward kernel of the state process
will also be degenerate. As a consequence, it is not possible to approximate the backward kernel
using the forward filter particles.

To illustrate how this difficulty can be remedied by a change of variables,
consider an \lgss model of the form
\begin{subequations}
  \label{eq:eval:lg_degen_sys}
  \begin{align}
    \begin{pmatrix}
      x_{t+1} \\ z_{t+1}
    \end{pmatrix} &=
    \begin{pmatrix}
      A_{11} & A_{12} \\ A_{21} & A_{22}
    \end{pmatrix}
    \begin{pmatrix}
      x_t \\ z_t
    \end{pmatrix} +
    \begin{pmatrix}
      v_t \\ 0
    \end{pmatrix}, & v_t&\sim\N(0,Q), \\
    y_t &= C
    \begin{pmatrix}
      x_t \\ z_t
    \end{pmatrix} + e_t, & e_t&\sim\N(0,R).
  \end{align}
\end{subequations}
Since the Gaussian process noise enters only on the first part of the state vector
(or, equivalently, the process noise covariance matrix is rank deficient)
the state transition kernel is degenerate. However, for the same reason, the state component $z_t$ is $\sigma(x_{1:t})$-measurable
and we can write $z_t = z_t(x_{1:t})$. Therefore, it is possible to rephrase \eqref{eq:eval:lg_degen_sys}
as a non-Markovian model with latent process given by $\{ x_t \}_{t\geq 1}$.

As a first illustration, we simulate $\T = 200$ samples from a fourth-order, single output system with poles at 
$-0.65$, $-0.12$, and $0.22 \pm 0.10i$. We let $\dim(x_t) = 1$ and $Q = R = 0.1$. For simplicity,
we assume that the system parameters are known and seek the joint smoothing distribution $p(x_{1:\T} \mid y_{1:\T})$.
In the non-Markovian formulation it is possible to apply backward-simulation-based methods, such as \pgas and \pgbs, as described in Section~\ref{sec:nonmarkov}.
The problem, however, is that the non-Markovianity gives rise to an $\Ordo(\T^2)$ computational complexity.
To obtain more practical inference algorithms we employ the weight truncation strategy \eqref{eq:nonmarkov:truncated_weights}.

First, we consider the coarse approximation $\ell = 1$.
We run \pgas and \pgbs, both with $\Np = 5$ particles for \thsnd{10} iterations (with the first \thsnd{1} discarded as burn-in).
We then compute running means of the latent variables $x_{1:\T}$ and,
from these, we compute the running root mean squared errors (RMSEs) $\epsilon_n$ \emph{relative to the true posterior means} (computed with a modified Bryson-Frazier smoother \cite{Bierman:1973}).
Hence, if no approximation would have been made, we would expect $\epsilon_n \rightarrow 0$, so any static error can be seen as the effect
of the truncation. The results for five independent runs are shown in Figure~\ref{fig:degen_4d_rmse}.
First, we note that both methods give accurate results. Still, the error for \pgas is significantly lower than for \pgbs.
For further comparison, we also run an \emph{untruncated} forward filter/backward simulator (\ffbsi)
particle smoother \cite{GodsillDW:2004}, using $\Np = \thsnd{10}$ particles and
$\Mp = \thsnd{1}$ backward trajectories (with a computational cost of $\Ordo(\Np\Mp\T^2)$). The resulting RMSE value is shown as a thick gray line in
Figure~\ref{fig:degen_4d_rmse}. This result suggest that \pgas can be a serious competitor to more ``classical'' particle smoothers,
even when there are no unknown parameters of the model. Already with $\ell=1$, \pgas outperforms \ffbsi in terms of accuracy and,
due to the fact that ancestor sampling allows us to use as few as $\Np = 5$ particles at each iteration, at a much lower computational cost.

\begin{figure}[ptb]
  \centering
  \includegraphics[width = 0.8\columnwidth]{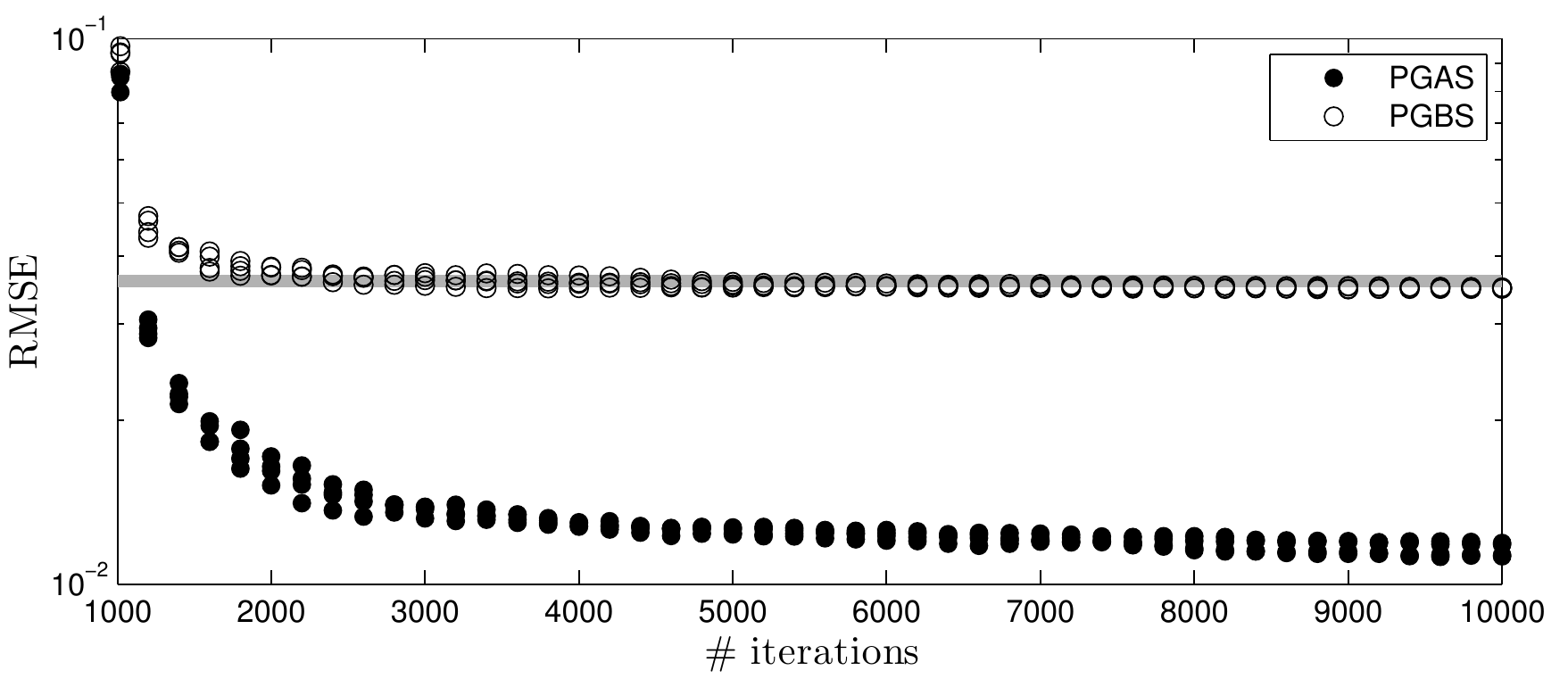}
  \caption{Running RMSEs for $x_{1:\T}$ for five independent runs of \pgas ($\bullet$) and \pgbs ($\circ$), respectively. The truncation level is set to $\ell = 1$.
    The thick gray line corresponds to a run of an untruncated \ffbsi particle smoother.}
  \label{fig:degen_4d_rmse}
\end{figure}

%%%%%%%%%%%%%%%%%%%%%%%%%%

To see how the samplers are affected by the choice of truncation level $\ell$ and by the mixing properties of the system,
we consider randomly generated systems of the form \eqref{eq:eval:lg_degen_sys} of different orders.
We generate 150 random systems, using the \matlab function
\verb+drss+ from the Control Systems Toolbox, with model orders 2, 5 and 20 (50 systems for each model order).
The number of outputs are taken as 1, 2 and 4 for the different model orders, respectively.
We consider different fixed truncation levels
($\ell\in\{1,2,3\}$ for 2nd order systems and $\ell\in\{1,5,10\}$ for 5th and 20th order systems),
as well as an adaptive level with $\upsilon = 0.1$ and $\tau = 10^{-2}$ (see Section~\ref{sec:nonmarkov:truncation}).
All other settings are as above.

Again, we compute the posterior means of $x_{1:\T}$ (discarding \thsnd{1} samples) and RMSE values relative the true posterior mean.
Box plots over the different systems are shown in Figure~\ref{fig:eval:degen_rndlgss}.
Since the process noise only enters on one of the state components, the mixing tends to deteriorate as we increase the model order.
Figure~\ref{fig:degen_Papprox} shows how the probability distributions on $\crange{1}{\Np}$ change as we increase the truncation level,
in two representative cases for a 5th and a 20th order system, respectively.
By using an adaptive level, we can obtain accurate results
for systems of different dimensions, without having to change any settings between the runs.

\begin{figure}[ptb]
  \centering
  \includegraphics[height = \boxplotheight]{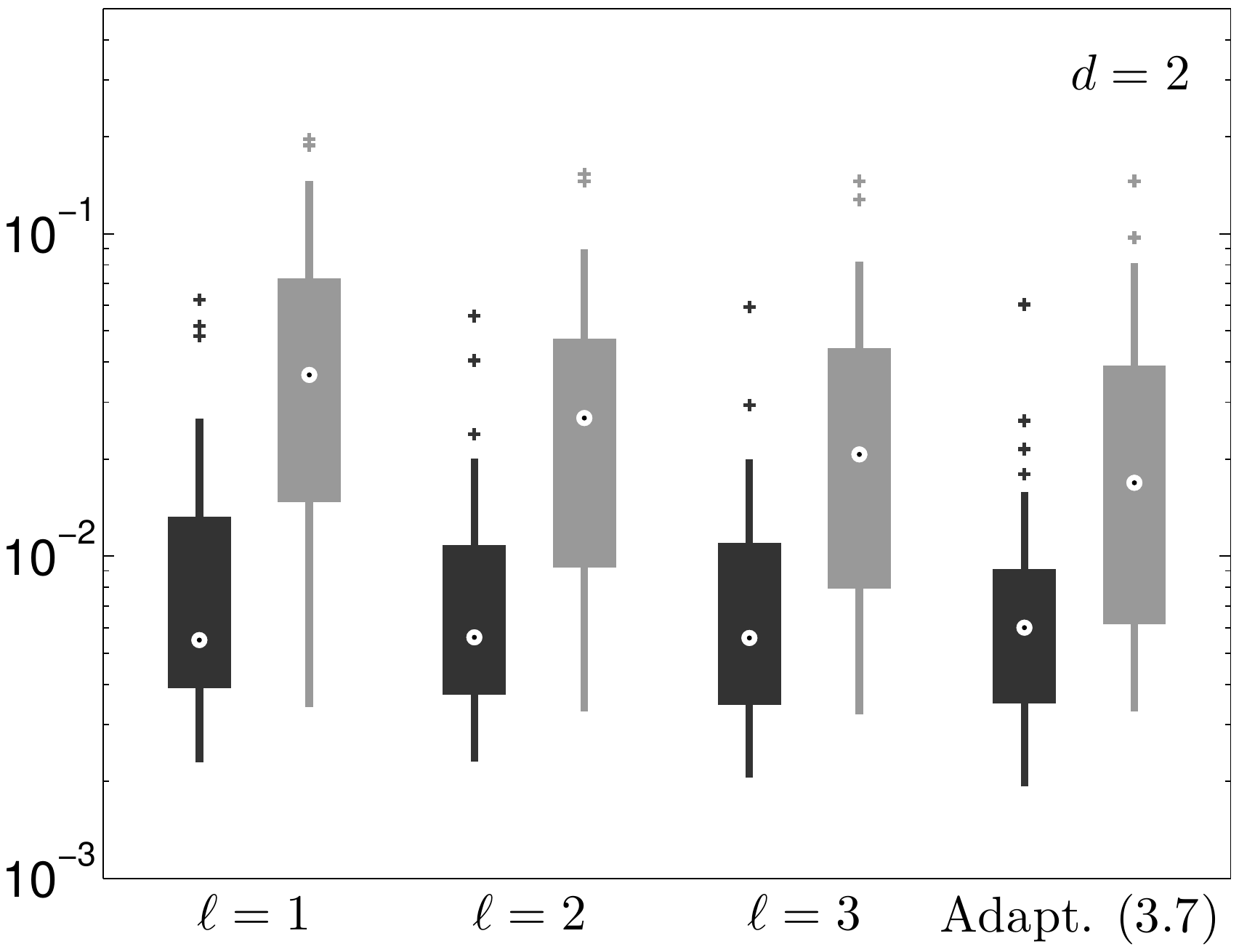}
  \includegraphics[height = \boxplotheight]{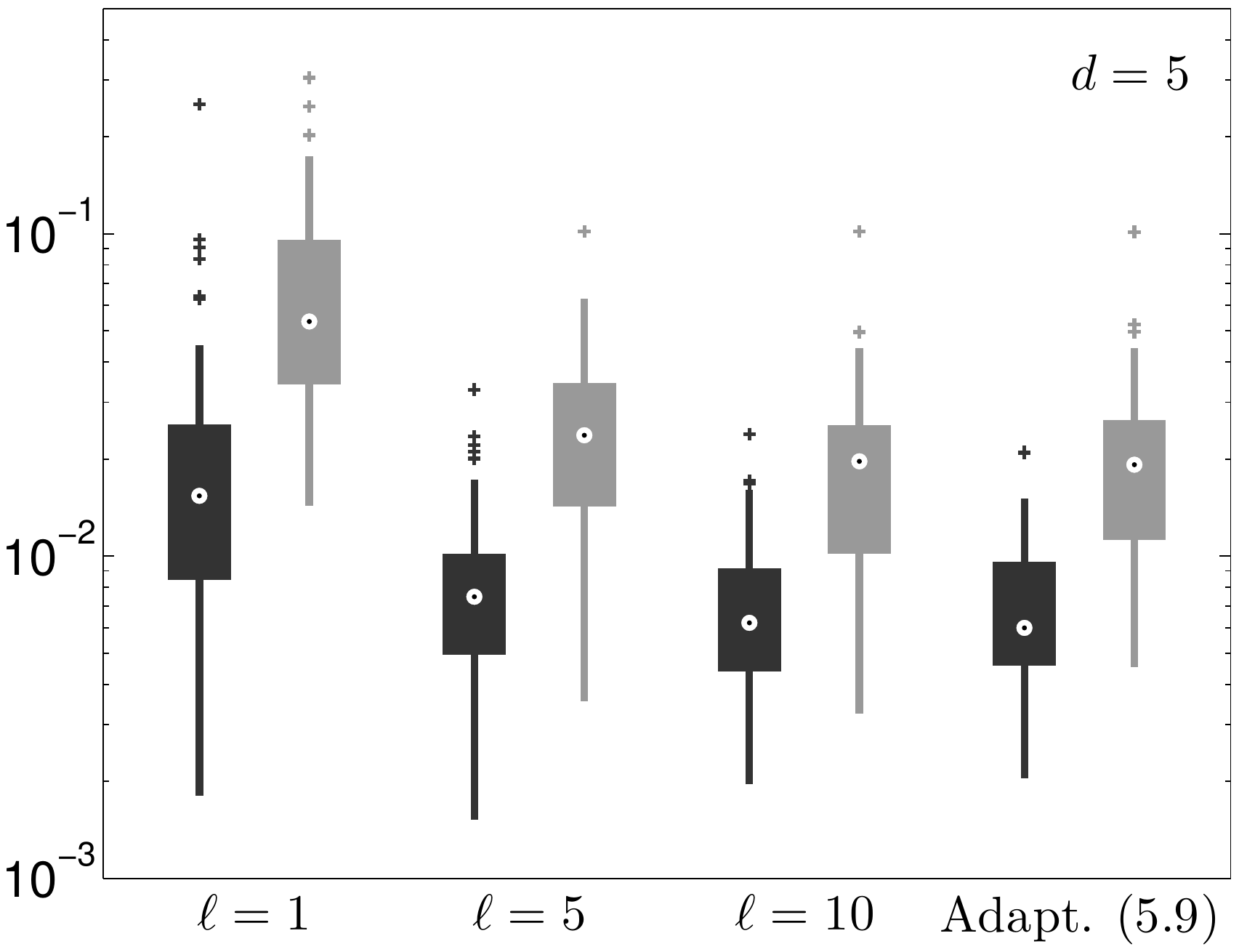}
  \includegraphics[height = \boxplotheight]{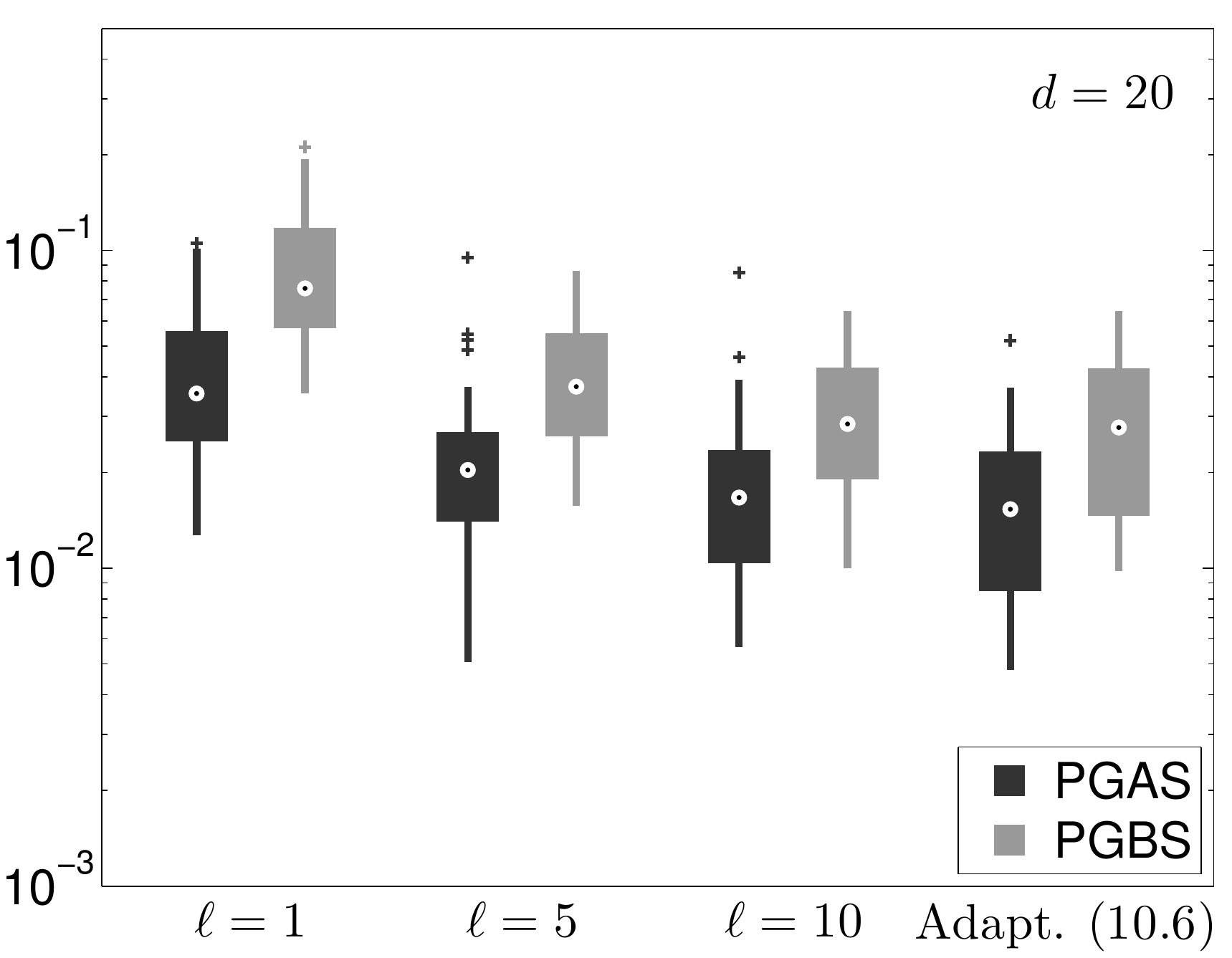}
  \caption{Box plots of the RMSE errors for \pgas (black) and \pgbs (gray),
    for 150 random systems of different dimensions $d$
    (left, $d = 2$; middle, $d = 5$; right, $d = 20$). Different values for the truncation level $\ell$ are considered. The rightmost boxes correspond
    to an adaptive truncation and the values in parentheses are the average truncation levels over all systems and MCMC iterations (the same for both methods).
    The dots within the boxes show the median errors.}
  \label{fig:eval:degen_rndlgss}
\end{figure}

\subsection{Epidemiological model}\label{sec:eval:flu}
As a final numerical illustration, we consider identification of an epidemiological
model using \pgas.
Seasonal influenza epidemics each year cause millions of severe illnesses and hundreds of thousands of deaths
world-wide \cite{GinsbergMPBSB:2009}. Furthermore, new strains of influenza viruses can possibly
cause pandemics with very severe effects on the public health. The ability to accurately predict disease activity
can enable early response to such epidemics, which in turn can reduce their impact.

We consider a susceptible/infected/recovered (SIR) model with environmental noise and seasonal
fluctuations \cite{KeelingR:2007,RasmussenRK:2011}. The model, specified by a stochastic differential equation,
is discretized according to the Euler-Maruyama method, yielding
\begin{subequations}
  \label{eq:epi_SIR}
  \begin{align}
    S_{t+\sampletime} &= S_t + \mu \populationsize \sampletime - \mu S_t \sampletime -
    \left( 1 + F v_t \right)\beta_t\frac{S_t}{\populationsize}I_t \sampletime, \\
    I_{t+\sampletime} &= I_t - (\gamma + \mu)I_t \sampletime + \left( 1 + F v_t \right)\beta_t\frac{S_t}{\populationsize}I_t \sampletime, \\
    R_{t+\sampletime} &= R_t + \gamma I_t \sampletime - \mu R_t \sampletime,
  \end{align}
\end{subequations}
where $v_t \sim \N(0, 1/\sqrt{\sampletime} )$ and $\sampletime$ is the sampling time.
Here, $S_t$, $I_t$ and $R_t$ represent the number of susceptible, infected and recovered individuals at time $t$ (months), respectively.
The total population size $\populationsize = 10^6$ and the host birth/death rate $\mu = 0.0012$ are assumed known.
The seasonally varying transmission rate is given by
$\beta_t = R_0(\gamma+\mu) (1+ \alpha \sin (2\pi t/12))$ where $R_0$ is the basic reproductive ratio, $\gamma$ is the rate of recovery
and $\alpha$ is the strength of seasonality.%

Furthermore, we consider an observation model which is inspired by the Google Flu Trends project \cite{GinsbergMPBSB:2009}.
The idea is to use the frequency of influenza-related search engine queries to infer knowledge of the dynamics
of the epidemic. Let $Q_k$ be the proportion of influenza-related queries counted during a time interval $(\Delta (k-1), \Delta k]$.
Following \cite{GinsbergMPBSB:2009}, we use a linear relationship between the log-odds of the relative query counts
and the log-odds of the proportion of infected individuals,
\begin{align}
  \label{eq:epi_obs}
  y_k &\triangleq \logit(Q_k) = \rho \logit( \bar I_k /\populationsize) + e_k, & e_k &\sim \N(0,\sigma^2),
\end{align}
where $\bar I_k$ is the mean value of $I_t$ during the time interval  $(\Delta (k-1), \Delta k]$
and $\logit(p) = \log(p/(1-p))$.
As in \cite{GinsbergMPBSB:2009} we consider weekly query counts, \ie, $\Delta = 7/30$ (assuming for simplicity that we have 30 days in each month).
Using this value of $\Delta$ as sampling time will, however, result in an overly large discretization errors. Instead, we sample the
model \eqref{eq:epi_SIR} $m=7$ times per week: $\sampletime = \Delta/m$.

In \cite{RasmussenRK:2011}, the particle marginal MH sampler \cite{AndrieuDH:2010}
is used to identify a similar SIR model, though with a different observation model.
A different Monte Carlo strategy, based on a particle filter with an augmented state space,
for identification of an SIR model is proposed in \cite{SkvortsovR:2011}.
We suggest to use the \pgas algorithm for joint state and parameter inference in the model \eqref{eq:epi_SIR}--\eqref{eq:epi_obs}.
However, there are two difficulties in applying \pgas  directly to this model.
Firstly, the transition kernel of the state process, as defined between consecutive observation time points $\Delta(k-1)$ and $\Delta k$,
is not available in closed form. Secondly, since the state is three-dimensional, whereas
the driving noise $v_t$ is scalar, the transition kernel is degenerate. To cope with these difficulties we (again)
suggest collapsing the model to the driving noise variables.
Let $V_k =
\begin{pmatrix}
  v_{\Delta(k-1)} &   v_{\Delta(k-1)+\sampletime} & \cdots & v_{\Delta k - \sampletime}
\end{pmatrix}^\T$. It follows that the model \eqref{eq:epi_SIR}--\eqref{eq:epi_obs} can be equivalently expressed as
the non-Markovian latent variable model,%
\begin{subequations}
  \label{eq:epi_SIR_collapsed}
  \begin{align}
    V_k &\sim \N(0, I_{m}/\sqrt{\sampletime}), \\
    y_k &\sim g_\theta(y_k \mid V_{1:k}),
  \end{align}
\end{subequations}
for some likelihood function $g_\theta$ (see \eqref{eq:epi_collapsed_likelihood}).
A further motivation for using this reformulation is that the latent variables $V_k$ are \emph{a priori}
independent of the model parameters $\theta$. This can result in a significant improvement in mixing
of the Gibbs sampler, in particular when there are strong dependencies between the system state
the parameters \cite{GolightlyW:2008,PapaspiliopoulosRS:2003}.

The parameters of the model are $\parameter = (\gamma, R_0, \alpha, F, \rho, \sigma)$,
with the true values given by
$\gamma = 3$, $R_0 = 10$, $\alpha = 0.16$, $F = 0.03$ , $\rho = 1.1$ and $\sigma  = 0.224$.
We place a normal-inverse-gamma
prior on the pair $( \rho, \sigma^2 )$, \ie, $p(\rho, \sigma^2) = \N(\rho ; \mu_\rho, c_\rho \sigma^2) \mathcal{IG}(\sigma^2 ; a_\sigma, b_\sigma)$.
The hyperparameters are chosen as $\mu_\rho = 1$, $c_\rho = 0.5$ and $a_\sigma = b_\sigma = 0.01$. For the remaining parameters,
we use improper flat priors on $\reals_+$.

We generate 8 years of data with weekly observations. The number of infected individuals $I_t$ over this time period is
shown in Figure~\ref{fig:epi_I}. The first half of the data batch is used for estimation of the model parameters. We
run \pgas with $\Np = 10$ for \thsnd{50} iterations (discarding the first \thsnd{10}).
For sampling the system parameters $(\gamma, R_0, \alpha, F)$, we use Metropolis-Hastings steps with a Gaussian random
walk proposal, tuned according to an initial trial run. For $(\rho, \sigma^2)$, we exploit the conjugacy
of the normal-inverse-gamma prior to the likelihood \eqref{eq:epi_obs} and sample
the variables from their true posterior. The innovation variables $V_{1:\T}$ are sampled
from the \pgas kernel by Algorithm~\ref{alg:pgas:pgas} (no truncation is used for the ancestor sampling weights).
Since the latter step is the computational bottleneck of the algorithm, we execute ten MH steps for $\theta$,
for each draw from the \pgas kernel.

It is worth pointing out that, while the sampler effectively targets the collapsed model \eqref{eq:epi_SIR_collapsed},
it is most straightforwardly implemented using the original state variables from \eqref{eq:epi_SIR}.
With $x_k = (S_{\Delta k}, I_{\Delta k}, R_{\Delta k})^\+$ we can simulate $x_{k+1}$ given $x_k$ according to \eqref{eq:epi_SIR}
which is used in the underlying particle filter. The innovation variables $V_k$ need only be taken into account
for the ancestor sampling step. Let $V_{1:\T}^\prime$ be the reference innovation trajectory.
To compute the ancestor sampling weights \eqref{eq:pgas:as_weights} we need to evaluate the ratios,
\begin{align}
  \frac{p_\theta( (V_{1:k-1}^i, V_{k:\T}^\prime), y_{1:\T})}{p_\theta(V_{1:k-1}^i, y_{1:k-1})}
  \propto \prod_{\ell = k}^\T g_\theta(y_{\ell} \mid V_{1:k-1}^i, V_{k:\ell}^\prime).
\end{align}
Using \eqref{eq:epi_obs}, the observation likelihood can be written as
\begin{align}
  \label{eq:epi_collapsed_likelihood}
  g_\theta(y_{\ell} \mid V_{1:k-1}^i, V_{k:\ell}^\prime) = \N(y_\ell \mid \rho \logit( \bar I_\ell\{ x_{k-1}^i, V_{k:\ell}^{\prime} \} /\populationsize), \sigma^2),
\end{align}
where $I_\ell\{ x_{k-1}^i, V_{k:\ell}^{\prime} \}$ is obtained by simulating the system \eqref{eq:epi_SIR}
from time ${\Delta(k-1)}$ to time $\Delta\ell$, initialized at $x_{k-1}^i$ and using the innovation sequence $ V_{k:\ell}^{\prime}$.

Histograms representing the estimated posterior parameter distributions are
shown in Figure~\ref{fig:epi_hist}. As can be seen, the true system parameters fall well within the credible regions.
Finally, the identified model is used to make one-month-ahead predictions of the disease activity for the subsequent four years,
as shown in Figure~\ref{fig:epi_I}. The predictions are made by sub-sampling the Markov chain and, for each sample, running
a particle filter on the validation data using 100 particles. As can be seen, we obtain an accurate prediction of the
disease activity, which falls within the estimated 95 \% credibility intervals, one month in advance.

\begin{figure}[ptb]
  \centering
  \includegraphics[width = 0.7\columnwidth]{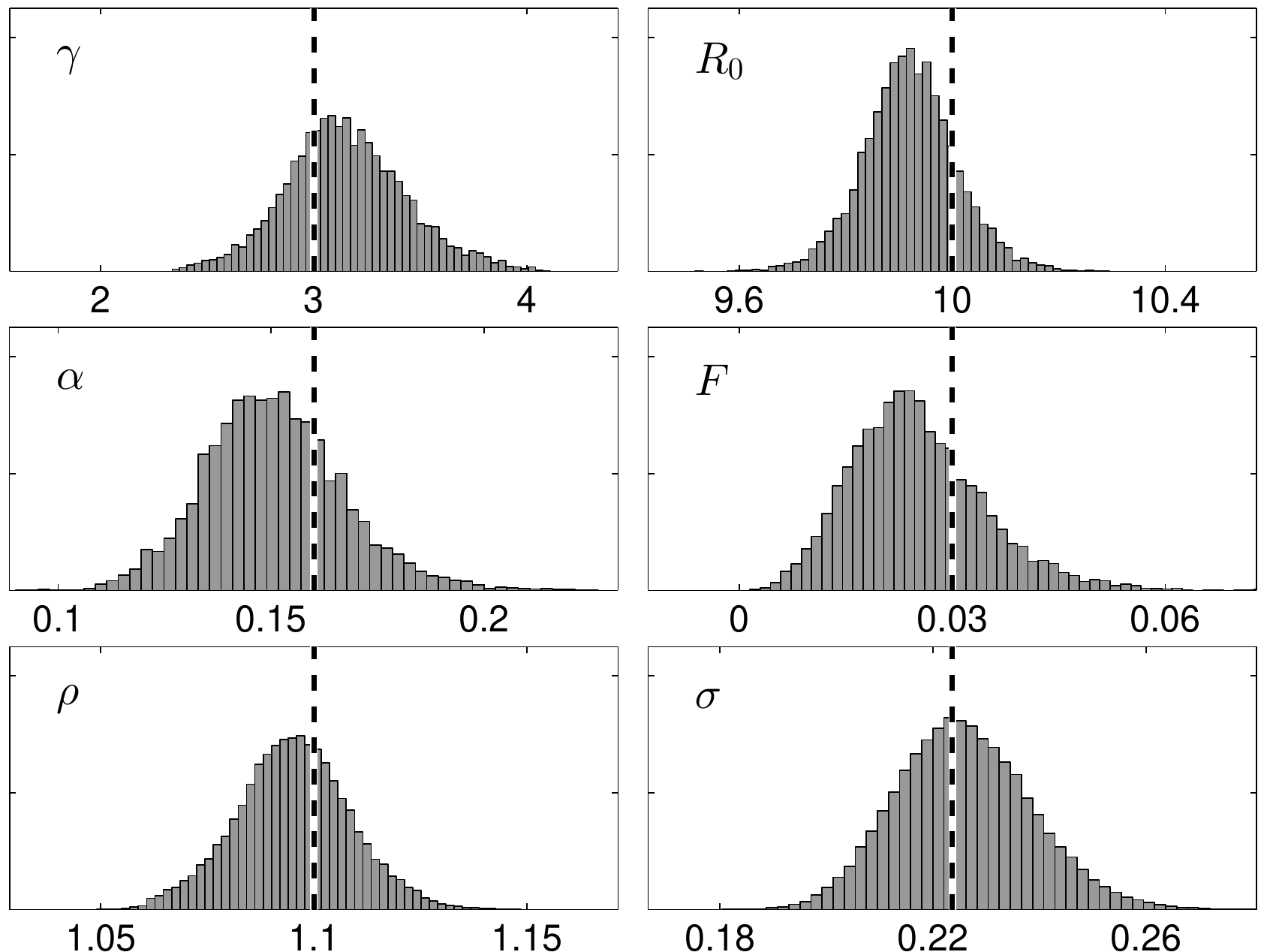}
  \caption{Posterior densities for the parameters of model \eqref{eq:epi_SIR}--\eqref{eq:epi_obs}.
    The true values are marked by vertical dashed lines.}
  \label{fig:epi_hist}
\end{figure}

\begin{figure*}[ptb]
  \centering
  \includegraphics[width = 1\linewidth]{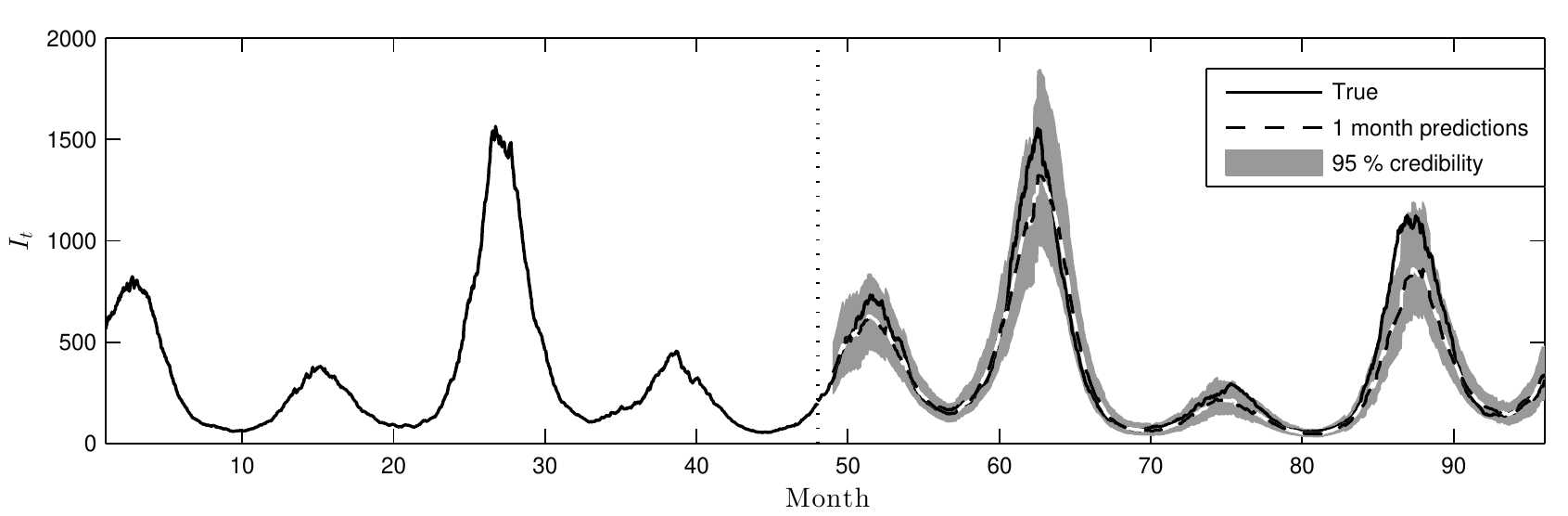}
  \caption{Disease activity (number of infected individuals $I_t$) over an eight year period. The first four years are used as estimation data, to find the unknown parameters of the model.  For the consecutive four years, one-month-ahead predictions are computed using the estimated model.}
  \label{fig:epi_I}
\end{figure*}

%%%%%%%%%%%%%%%%%%%%%%%
%     NEW SECTION     %
%%%%%%%%%%%%%%%%%%%%%%%
\section{Discussion}\label{sec:discussion}%
\pgas is a novel approach to \pmcmc that provides the statistician with an off-the-shelf class
of Markov kernels which can be used to simulate, for instance, the typically high-dimensional and highly autocorrelated state trajectory
in a state-space model. This opens up for using \pgas as a key component in different inference algorithms,
enabling both Bayesian and frequentist parameter inference as well as state inference.
However, \pgas by no means limited to inference in state-space models. Indeed,
we believe that the method can be particularly useful for models with more complex dependencies,
such as non-Markovian, nonparametric, and graphical models.

The \pgas Markov kernels are built upon two main ideas. First, by conditioning the underlying \smc sampler
on a reference trajectory the correct stationary distribution of the kernel is enforced.
Second, \emph{ancestor sampling} enables movement around the reference trajectory which drastically improves
the mixing of the sampler. In particular, we have shown empirically that ancestor sampling makes
the mixing of the \pgas kernels robust to a small number of particles as well as to large data records.

Ancestor sampling is basically a way of exploiting backward simulation ideas without needing an explicit backward
pass. Compared to \pgbs, a conceptually similar method that does require an explicit backward pass, \pgas has
several advantages, most notably for inference in non-Markovian models. When using the proposed truncation
of the backward weights, we have found \pgas to be more robust to the approximation error than \pgbs, yielding
up to an order-of-magnitude improvement in accuracy.
An interesting topic for future work is to further investigate the effect on these samplers
by errors in the backward weights, whether these errors arise from a truncation or some other
approximation of the transition density function.
It is also worth pointing out that for non-Markovian model \pgas is simpler to implement than \pgbs as it requires less bookkeeping.
It can also be more memory efficient since it does not require storage of intermediate
quantities that are needed for a separate backward simulation pass. See \cite{JacobMR:2013} for a related
discussion on path storage in the particle filter.

Other directions for future work include further analysis of the ergodicity of \pgas.
While the established uniform ergodicity result is encouraging, it does not provide
information about how fast the mixing rate improves with the number of particles. Finding
informative rates with an explicit dependence on $\Np$ is an interesting, though challenging,
topic for future work. It would also be interesting to further investigate empirically the convergence rate of \pgas
for different settings, such as the number of particles, the amount of data, and the dimension of the latent process.

\appendix

\section{Proofs}\label{app:proof}

\subsection{Proof of Proposition~\ref{prop:equiv}}%
In this appendix we prove Proposition~\ref{prop:equiv}.
For improved readability we provide the complete proof,
though it should be noted that the main part is due to \cite{OlssonR:2011}.
For ease of notation, we write $\E$ for $\E_{\theta,x_{1:\T}^\prime}$.
First, note that
for a bootstrap proposal kernel, the weight function \eqref{eq:smc:weightfunction} is given by
$\wf{\theta}{t}(x_t) = g_\theta(y_t \mid x_t)$, \ie, it depend only on the current state and not on its ancestor.
As a consequence, the law of the forward particle system is independent of the ancestor variables
$\{ a_t^\Np \}_{t=2}^\T$, meaning that the particle systems (excluding $\{ a_t^\Np \}_{t=2}^\T$) are
equally distributed for \pgas and for \pgbs.

Let $B\in\sigmaX^\T$ be a measurable rectangle: $B = \times_{t=1}^\T B_t$
with $B_t \in \sigmaX$ for $t = \range{1}{\T}$. Then,
\begin{align*}
  \kPGAS(x_{1:\T}^\prime, B) &= \E \left[ \prod_{t=1}^\T \I_{B_t}(x_t^{b_t}) \right], &&\text{and} &
  \kPGBS(x_{1:\T}^\prime, B) &= \E \left[ \prod_{t=1}^\T \I_{B_t}(x_t^{j_t}) \right].
\end{align*}
Since the measurable rectangles form a $\pi$-system generating $\sigmaX^\T$, it is by the $\pi$-$\lambda$ theorem sufficient to show that
$\E [ h(x_{1:\T}^{b_{1:\T}}) ] = \E [ h(x_{1:\T}^{j_{1:\T}}) ]$
for all bounded, multiplicative functionals, $h(x_{1:\T}) = \prod_{t=1}^\T h_t(x_t)$. As in \cite{OlssonR:2011}, we establish this result
by induction. Hence, for $t<\T$, assume that
\begin{align}
  \E \left[ \prod_{s=t+1}^\T h_{s}(x_s^{b_s}) \right]
  = \E \left[ \prod_{s=t+1}^\T h_s(x_s^{j_s}) \right].
\end{align}
For $t = \T-1$, the induction hypothesis holds since $b_\T$ and $j_\T$ are equally distributed
(both are drawn from the discrete distribution induced by the weights $\{w_\T^i \}_{i=1}^\Np$).
Let
\begin{align}
  \nonumber
  \Lambda_t( x_{t+1}^{j_{t+1}}, h) &\triangleq \E\left[ h( x_{t}^{j_{t}}) \mid  x_{t+1}^{j_{t+1}}\right]
  = \E\left[ \E\left[ h( x_{t}^{j_{t}}) \mid \xx_t,  x_{t+1}^{j_{t+1}}\right]  \mid  x_{t+1}^{j_{t+1}}\right] \\
  &= \E\left[ \sum_{i=1}^\Np  h( x_{t}^{i}) \frac{w_t^i f_{\theta}(x_{t+1}^{j_{t+1}} \mid x_t^i) }{ \normsum{l} w_t^l f_{\theta}(x_{t+1}^{j_{t+1}} \mid x_t^l) }  \mid  x_{t+1}^{j_{t+1}} \right],
\end{align}
where we recall that $w_t^i = \wf{\theta}{t}(x_{t}^i)$ and where the last equality follows from~\eqref{eq:equiv:bsi}.
Consider,
\begin{align}
  \label{eq:equiv:outerexp}
  \E \left[ \prod_{s=t}^\T h_{s}(x_s^{b_s}) \right] = 
  \E \left[ \E\left[ h_t(x_t^{b_t}) \mid x_{t+1:\T}^{b_{t+1:\T}}, b_{t+1:\T} \right] \prod_{s=t+1}^\T h_{s}(x_s^{b_s}) \right].
\end{align}
Using the Markov property of the generated particle system and the tower property of conditional expectation, we have
\begin{align}
  \label{eq:equiv:innerexp}
  \E\left[ h_t(x_t^{b_t}) \mid x_{t+1:\T}^{b_{t+1:\T}}, b_{t+1:\T} \right] =
  \E\left[ \E\left[ h_t(x_t^{b_t}) \mid \xx_t, x_{t+1}^{b_{t+1}}, b_{t+1} \right] \mid  x_{t+1}^{b_{t+1}}, b_{t+1} \right].
\end{align}
Recall that $b_{t} = a_{t+1}^{b_{t+1}}$. Consider first the case $b_{t+1} < \Np$. From \eqref{eq:smc:joint_proposal},
we have that $\Prb( b_t = i \mid \xx_t) \propto w_t^i$ and $x_{t+1}^{b_{t+1}} \mid x_t^{b_t} \sim f_\theta(\,\cdot \mid x_t^{b_t})$.
If follows from Bayes' theorem that $\Prb( b_t = i \mid \xx_t, x_{t+1}^{b_{t+1}}) \propto w_t^i f_\theta(x_{t+1}^{b_{t+1}} \mid x_t^{b_t})$.
However, by the ancestor sampling procedure (Algorithm~2, line~\ref{row:atN}), the same expression holds also for $b_{t+1} = \Np$.
We can thus write~\eqref{eq:equiv:innerexp} as
\begin{align}
  \nonumber
  \E\left[ h_t(x_t^{b_t}) \mid x_{t+1:\T}^{b_{t+1:\T}}, b_{t+1:\T} \right] &=
  \E\left[ \sum_{i=1}^\Np  h_t( x_{t}^{i}) \frac{w_t^i f_{\theta}(x_{t+1}^{b_{t+1}} \mid x_t^i) }{ \normsum{l} w_t^l f_{\theta}(x_{t+1}^{b_{t+1}} \mid x_t^l) }  \mid  x_{t+1}^{b_{t+1}}, b_{t+1} \right] \\
  &= \Lambda_t(x_{t+1}^{b_{t+1}}, h_t),
\end{align}
Hence, since the function $x_{t+1} \mapsto \Lambda_t(x_{t+1}, h_t)$ is bounded, we can use the induction hypothesis to write
\eqref{eq:equiv:outerexp} as
\begin{align*}
  \E &\left[ \prod_{s=t}^\T h_{s}(x_s^{b_s}) \right] = \E \left[  \Lambda_t(x_{t+1}^{b_{t+1}}, h_t) \prod_{s=t+1}^\T h_{s}(x_s^{b_s}) \right]
  = \E \left[  \Lambda_t(x_{t+1}^{j_{t+1}}, h_t) \prod_{s=t+1}^\T h_{s}(x_s^{j_s}) \right] \\
  &\hspace{1em}= \E \left[  \E\left[ h_t(x_t^{j_t}) \mid x_{t+1:\T}^{j_{t+1:\T}}, j_{t+1:\T} \right]  \prod_{s=t+1}^\T h_{s}(x_s^{j_s}) \right]
  = \E \left[ \prod_{s=t}^\T h_{s}(x_s^{j_s}) \right].
\end{align*}
\hfill$\square$

\subsection{Proof of Proposition~\ref{prop:truncation_kld}}
  With $M = \T - t + 1$ and $w(k) = w_{t-1}^k$, the distributions of interest are given by
  \begin{align*}
    \asprob(k) &= \frac{w(k)\prod_{s = 1}^M h_s(k)} {\normsum{l} w(l)\prod_{s = 1}^M h_s(l)} & &\text{ and } &
    \asprobtrunc(k) &= \frac{w(k)\prod_{s = 1}^\ell h_s(k)} {\normsum{l} w(l)\prod_{s = 1}^\ell h_s(l)},
  \end{align*}
  respectively. Let
  $\varepsilon_s \triangleq \max_{k,l} \left( h_s(k) / h_s(l)-1 \right) \leq A\exp(-cs)$
  and consider
  \begin{align*}
    \left( \sum_l w(l) \prod_{s=1}^\ell h_s(l) \right) \prod_{s = \ell+1}^M h_s(k) &\leq 
    \sum_l  \left( w(l) \prod_{s=1}^\ell h_s(l) \prod_{s = \ell+1}^M h_s(l)(1+\varepsilon_s) \right) \\
    &= \left(  \sum_l  w(l) \prod_{s=1}^M h_s(l)  \right) \prod_{s = \ell+1}^M (1+\varepsilon_s) .
  \end{align*}
  It follows that the KL divergence is bounded according to,
  \begin{align*}
    \nonumber
    \KLD&(\asprob \| \asprobtrunc) = \sum_k \asprob(k) \log \frac{\asprob(k)}{\asprobtrunc(k)} \\
    \nonumber
    &= \sum_k \asprob(k) \log \left( \frac{ \prod_{s = \ell+1}^M h_s(k) \left( \sum_l w(l) \prod_{s=1}^\ell h_s(l) \right) }{  \sum_l w(l) \prod_{s=1}^M h_s(l) } \right) \\
    \nonumber
    &\leq \sum_k \asprob(k) \sum_{s = \ell+1}^M \log(1+\varepsilon_s) \leq \sum_{s = \ell+1}^M \varepsilon_s \leq A \sum_{s = \ell+1}^M \exp(-cs) \\
    &= A \frac{e^{-c(\ell+1)} - e^{-c(M+1)}}{1-e^{-c}}.
  \end{align*}
  \hfill$\square$

\bibliographystyle{IEEEtran}
\bibliography{references}

\begin{thebibliography}{10}
\providecommand{\url}[1]{#1}
\csname url@samestyle\endcsname
\providecommand{\newblock}{\relax}
\providecommand{\bibinfo}[2]{#2}
\providecommand{\BIBentrySTDinterwordspacing}{\spaceskip=0pt\relax}
\providecommand{\BIBentryALTinterwordstretchfactor}{4}
\providecommand{\BIBentryALTinterwordspacing}{\spaceskip=\fontdimen2\font plus
\BIBentryALTinterwordstretchfactor\fontdimen3\font minus
  \fontdimen4\font\relax}
\providecommand{\BIBforeignlanguage}[2]{{%
\expandafter\ifx\csname l@#1\endcsname\relax
\typeout{** WARNING: IEEEtran.bst: No hyphenation pattern has been}%
\typeout{** loaded for the language `#1'. Using the pattern for}%
\typeout{** the default language instead.}%
\else
\language=\csname l@#1\endcsname
\fi
#2}}
\providecommand{\BIBdecl}{\relax}
\BIBdecl

\bibitem{DoucetJ:2011}
A.~Doucet and A.~Johansen, ``A tutorial on particle filtering and smoothing:
  Fifteen years later,'' in \emph{The Oxford Handbook of Nonlinear Filtering},
  D.~Crisan and B.~Rozovskii, Eds.\hskip 1em plus 0.5em minus 0.4em\relax
  Oxford University Press, 2011.

\bibitem{DelMoralDJ:2006}
P.~Del~Moral, A.~Doucet, and A.~Jasra, ``Sequential {M}onte {C}arlo samplers,''
  \emph{Journal of the Royal Statistical Society: Series {B}}, vol.~68, no.~3,
  pp. 411--436, 2006.

\bibitem{RobertC:2004}
C.~P. Robert and G.~Casella, \emph{{M}onte {C}arlo Statistical Methods}.\hskip
  1em plus 0.5em minus 0.4em\relax Springer, 2004.

\bibitem{Liu:2001}
J.~S. Liu, \emph{{M}onte {C}arlo Strategies in Scientific Computing}.\hskip 1em
  plus 0.5em minus 0.4em\relax Springer, 2001.

\bibitem{LindstenS:2013}
F.~Lindsten and T.~B. Sch\"on, ``Backward simulation methods for {M}onte
  {C}arlo statistical inference,'' \emph{Foundations and Trends in Machine
  Learning}, vol.~6, no.~1, pp. 1--143, 2013.

\bibitem{AndrieuDH:2010}
C.~Andrieu, A.~Doucet, and R.~Holenstein, ``Particle {M}arkov chain {M}onte
  {C}arlo methods,'' \emph{Journal of the Royal Statistical Society: Series B},
  vol.~72, no.~3, pp. 269--342, 2010.

\bibitem{VrugtBDS:2013}
J.~A. Vrugt, J.~F. ter Braak, C.~G.~H. Diks, and G.~Schoups, ``Hydrologic data
  assimilation using particle {M}arkov chain {M}onte {C}arlo simulation:
  {T}heory, concepts and applications,'' \emph{Advances in Water Resources},
  vol.~51, pp. 457--478, 2013.

\bibitem{PittSGK:2012}
M.~K. Pitt, R.~S. Silva, P.~Giordani, and R.~Kohn, ``On some properties of
  {M}arkov chain {M}onte {C}arlo simulation methods based on the particle
  filter,'' \emph{Journal of Econometrics}, vol. 171, pp. 134--151, 2012.

\bibitem{GolightlyW:2011}
A.~Golightly and D.~J. Wilkinson, ``Bayesian parameter inference for stochastic
  biochemical network models using particle {M}arkov chain {M}onte {C}arlo,''
  \emph{Interface Focus}, vol.~1, no.~6, pp. 807--820, 2011.

\bibitem{RasmussenRK:2011}
D.~A. Rasmussen, O.~Ratmann, and K.~Koelle, ``Inference for nonlinear
  epidemiological models using genealogies and time series,'' \emph{PLoS Comput
  Biology}, vol.~7, no.~8, 2011.

\bibitem{ChopinS:2013}
N.~Chopin and S.~S. Singh, ``On the particle {G}ibbs sampler,'' arXiv.org,
  arXiv:1304.1887, Apr. 2013.

\bibitem{WhiteleyAD:2010}
N.~Whiteley, C.~Andrieu, and A.~Doucet, ``Efficient {B}ayesian inference for
  switching state-space models using discrete particle {M}arkov chain {M}onte
  {C}arlo methods,'' Bristol Statistics Research Report 10:04, Tech. Rep.,
  2010.

\bibitem{LindstenS:2012}
F.~Lindsten and T.~B. Sch\"{o}n, ``On the use of backward simulation in the
  particle {G}ibbs sampler,'' in \emph{Proceedings of the 37th {IEEE}
  International Conference on Acoustics, Speech and Signal Processing
  ({ICASSP})}, Kyoto, Japan, Mar. 2012.

\bibitem{LindstenJS:2012}
F.~Lindsten, M.~I. Jordan, and T.~B. Sch\"on, ``Ancestor sampling for particle
  {G}ibbs,'' in \emph{Advances in Neural Information Processing Systems
  ({NIPS}) 25}, P.~Bartlett, F.~C.~N. Pereira, C.~J.~C. Burges, L.~Bottou, and
  K.~Q. Weinberger, Eds., 2012, pp. 2600--2608.

\bibitem{PittS:1999}
M.~K. Pitt and N.~Shephard, ``Filtering via simulation: Auxiliary particle
  filters,'' \emph{Journal of the American Statistical Association}, vol.~94,
  no. 446, pp. 590--599, 1999.

\bibitem{ChesneyS:1989}
M.~Chesney and L.~Scott, ``Pricing european currency options: A comparison of
  the modified {B}lack-{S}choles model and a random variance model,'' \emph{The
  Journal of Financial and Quantitative Analysis}, vol.~24, no.~3, pp.
  267--284, 1989.

\bibitem{MelinoT:1990}
A.~Melino and S.~M. Turnbull, ``Pricing foreign currency options with
  stochastic volatility,'' \emph{Journal of Econometrics}, vol.~45, no. 1--2,
  pp. 239--265, 1990.

\bibitem{DykP:2008}
D.~A.~V. Dyk and T.~Park, ``Partially collapsed {G}ibbs samplers: Theory and
  methods,'' \emph{Journal of the American Statistical Association}, vol. 103,
  no. 482, pp. 790--796, 2008.

\bibitem{Tierney:1994}
L.~Tierney, ``{M}arkov chains for exploring posterior distributions,''
  \emph{The Annals of Statistics}, vol.~22, no.~4, pp. 1701--1728, 1994.

\bibitem{JongS:1995}
P.~de~Jong and N.~Shephard, ``The simulation smoother for time series models,''
  \emph{Biometrika}, vol.~82, no.~2, pp. 339--350, 1995.

\bibitem{DempsterLR:1977}
A.~Dempster, N.~Laird, and D.~Rubin, ``Maximum likelihood from incomplete data
  via the {EM} algorithm,'' \emph{Journal of the {R}oyal {S}tatistical
  {S}ociety, {S}eries {B}}, vol.~39, no.~1, pp. 1--38, 1977.

\bibitem{McLachlanK:2008}
G.~McLachlan and T.~Krishnan, \emph{The {EM} Algorithm and Extensions},
  2nd~ed., ser. Whiley Series in Probability and Statistics.\hskip 1em plus
  0.5em minus 0.4em\relax New York, USA: John Wiley \& Sons, 2008.

\bibitem{WeiT:1990}
G.~C.~G. Wei and M.~A. Tanner, ``A {M}onte {C}arlo implementation of the {EM}
  algorithm and the poor man's data augmentation algorithms,'' \emph{Journal of
  the American Statistical Association}, vol.~85, no. 411, pp. 699--704, 1990.

\bibitem{DelyonLM:1999}
B.~Delyon, M.~Lavielle, and E.~Moulines, ``Convergence of a stochastic
  approximation version of the {EM} algorithm,'' \emph{The Annals of
  Statistics}, vol.~27, no.~1, pp. 94--128, 1999.

\bibitem{BenvenisteMP:1990}
A.~Benveniste, M.~M{\'e}tivier, and P.~Priouret, \emph{Adaptive algorithms and
  stochastic approximations}.\hskip 1em plus 0.5em minus 0.4em\relax New York,
  USA: Springer-Verlag, 1990.

\bibitem{AndrieuMP:2005}
C.~Andrieu, E.~Moulines, and P.~Priouret, ``Stability of stochastic
  approximation under verifiable conditions,'' \emph{SIAM Journal on Control
  and Optimization}, vol.~44, no.~1, pp. 283--312, 2005.

\bibitem{Lindsten:2013}
F.~Lindsten, ``An efficient stochastic approximation {EM} algorithm using
  conditional particle filters,'' in \emph{Proceedings of the 38th {IEEE}
  International Conference on Acoustics, Speech and Signal Processing
  ({ICASSP})}, Vancouver, Canada, May 2013.

\bibitem{GodsillDW:2004}
S.~J. Godsill, A.~Doucet, and M.~West, ``{M}onte {C}arlo smoothing for
  nonlinear time series,'' \emph{Journal of the American Statistical
  Association}, vol.~99, no. 465, pp. 156--168, Mar. 2004.

\bibitem{Whiteley:2010}
N.~Whiteley, ``Discussion on {P}article {M}arkov chain {M}onte {C}arlo
  methods,'' Journal of the Royal Statistical Society: Series {B}, 72(3), p
  306--307, 2010.

\bibitem{OlssonR:2011}
J.~Olsson and T.~Ryd\'en, ``{R}ao-{B}lackwellization of particle {M}arkov chain
  {M}onte {C}arlo methods using forward filtering backward sampling,''
  \emph{{IEEE} Transactions on Signal Processing}, vol.~59, no.~10, pp.
  4606--4619, 2011.

\bibitem{HjortHMW:2010}
N.~L. Hjort, C.~Holmes, P.~M{\"u}ller, and S.~G. Walker, Eds., \emph{Bayesian
  Nonparametrics}.\hskip 1em plus 0.5em minus 0.4em\relax Cambridge University
  Press, 2010.

\bibitem{TehJBB:2006}
Y.~W. Teh, M.~I. Jordan, M.~J. Beal, and D.~M. Blei, ``Hierarchical {D}irichlet
  processes,'' \emph{Journal of the American Statistical Association}, vol.
  101, no. 476, pp. 1566--1581, 2006.

\bibitem{EscobarW:1995}
M.~D. Escobar and M.~West, ``Bayesian density estiamtion and inference using
  mixtures,'' \emph{Journal of the American Statistical Association}, vol.~90,
  no. 430, pp. 577--588, 1995.

\bibitem{RasmussenW:2006}
C.~E. Rasmussen and C.~K.~I. Williams, \emph{Gaussian Processes for Machine
  Learning}.\hskip 1em plus 0.5em minus 0.4em\relax {MIT} Press, 2006.

\bibitem{TurnerD:2010}
R.~Turner and C.~E. Deisenroth, M. P.~Rasmussen, ``State-space inference and
  learning with {G}aussian processes,'' in \emph{Proceedings of the 13th
  International Conference on Artificial Intelligence and Statistics}, 2010.

\bibitem{FrigolaLSR:2013}
R.~Frigola, F.~Lindsten, T.~B. Sch\"on, and C.~E. Rasmussen, ``{B}ayesian
  inference and learning in {G}aussian process state-space models with particle
  {MCMC},'' in \emph{Advances in Neural Information Processing Systems ({NIPS})
  26 (to appear)}, Dec. 2013.

\bibitem{ChenL:2000}
R.~Chen and J.~S. Liu, ``Mixture {K}alman filters,'' \emph{Journal of the Royal
  Statistical Society: Series {B}}, vol.~62, no.~3, pp. 493--508, 2000.

\bibitem{LindstenBGS:2013}
F.~Lindsten, P.~Bunch, S.~J. Godsill, and T.~B. Sch\"on, ``Rao-{B}lackwellized
  particle smoothers for mixed linear/nonlinear state-space models,'' in
  \emph{Proceedings of the 38th {IEEE} International Conference on Acoustics,
  Speech and Signal Processing ({ICASSP})}, Vancouver, Canada, May 2013.

\bibitem{GanderS:2007}
M.~P.~S. Gander and D.~A. Stephens, ``Stochastic volatility modelling in
  continuous time with general marginal distributions: {I}nference, prediction
  and model selection,'' \emph{Journal of Statistical Planning and Inference},
  vol. 137, no.~10, pp. 3068--3081, 2007.

\bibitem{FearnheadPR:2008}
P.~Fearnhead, O.~Papaspiliopoulos, and G.~O. Roberts, ``Particle filters for
  partially observed diffusions,'' \emph{Journal of the Royal Statistical
  Society: Series {B}}, vol.~70, no.~4, pp. 755--777, 2008.

\bibitem{GolightlyW:2008}
A.~Golightly and D.~J. Wilkinson, ``Bayesian inference for nonlinear
  multivariate diffusion models observed with error,'' \emph{Computational
  Statistics \& Data Analysis}, vol.~52, no.~3, pp. 1674--1693, 2008.

\bibitem{MurrayJP:2012}
L.~M. Murray, E.~M. Jones, and J.~Parslow, ``On collapsed state-space models
  and the particle marginal {M}etropolis-{H}astings sampler,'' arXiv.org,
  arXiv:1202.6159v1, Feb. 2012.

\bibitem{RisticAG:2004}
B.~Ristic, S.~Arulampalam, and N.~Gordon, \emph{Beyond the {K}alman filter:
  particle filters for tracking applications}.\hskip 1em plus 0.5em minus
  0.4em\relax London, UK: Artech House, 2004.

\bibitem{GustafssonGBFJKN:2002}
F.~Gustafsson, F.~Gunnarsson, N.~Bergman, U.~Forssell, J.~Jansson, R.~Karlsson,
  and P.-J. Nordlund, ``Particle filters for positioning, navigation, and
  tracking,'' \emph{{IEEE} Transactions on Signal Processing}, vol.~50, no.~2,
  pp. 425--437, 2002.

\bibitem{KollerF:2009}
D.~Koller and N.~Friedman, \emph{Probabilistic Graphical Models}.\hskip 1em
  plus 0.5em minus 0.4em\relax {MIT} Press, 2009.

\bibitem{Liu:1996a}
J.~S. Liu, ``Peskun's theorem and a modified discrete-state {G}ibbs sampler,''
  \emph{Biometrika}, vol.~83, no.~3, pp. 681--682, 1996.

\bibitem{Bierman:1973}
G.~J. Bierman, ``Fixed interval smoothing with discrete measurements,''
  \emph{International Journal of Control}, vol.~18, no.~1, pp. 65--75, 1973.

\bibitem{GinsbergMPBSB:2009}
J.~Ginsberg, M.~H. Mohebbi, R.~S. Patel, L.~Brammer, M.~S. Smolinski, and
  L.~Brilliant, ``Detecting influenza epidemics using search engine query
  data,'' \emph{Nature}, vol. 457, pp. 1012--1014, 2009.

\bibitem{KeelingR:2007}
M.~Keeling and P.~Rohani, \emph{Modeling Infectious Diseases in Humans and
  Animals}.\hskip 1em plus 0.5em minus 0.4em\relax Princeton University Press,
  2007.

\bibitem{SkvortsovR:2011}
A.~Skvortsov and B.~Ristic, ``Monitoring and prediction of an epidemic outbreak
  using syndromic observations,'' arXiv.org, arXiv:1110.4696v1, Oct. 2011.

\bibitem{PapaspiliopoulosRS:2003}
O.~Papaspiliopoulos, G.~O. Roberts, and M.~Sk{\"o}ld, ``Non-centered
  parameterisations for hierarchical models and data augmentation,'' in
  \emph{Bayesian Statistics 7}, J.~M. Bernardo, M.~J. Bayarri, J.~O. Berger,
  A.~P. Dawid, D.~Heckerman, A.~F.~M. Smith, and M.~West, Eds.\hskip 1em plus
  0.5em minus 0.4em\relax Oxford University Press, 2003, pp. 307--326.

\bibitem{JacobMR:2013}
P.~E. Jacob, L.~Murray, and S.~Rubenthaler, ``Path storage in the particle
  filter,'' arXiv.org, arXiv:1307.3180, Jul. 2013.

\end{thebibliography}

\end{document}